%% file: samplepaper.tex
\documentclass[runningheads]{llncs}

\usepackage{graphicx}
\usepackage{tcolorbox}
\AtBeginDocument{%
  \providecommand\BibTeX{{%
    \normalfont B\kern-0.5em{\scshape i\kern-0.25em b}\kern-0.8em\TeX}}}

\usepackage{graphicx}
\usepackage{wrapfig}
\usepackage{booktabs} 

\usepackage{url}
\usepackage{algorithm}
\usepackage[noend]{algpseudocode}
\usepackage{amssymb,amsmath}
\usepackage{subcaption}
\usepackage{graphicx}
\usepackage{comment}
\usepackage{thmtools, thm-restate}
\usepackage{multirow}

\usepackage[graphicx]{realboxes}

\usepackage{xcolor} 
\definecolor{ocre}{RGB}{243,102,25} 

\usepackage{avant} 
\usepackage{mathptmx} 

\usepackage{microtype} 
\usepackage[utf8x]{inputenc} 
\usepackage[T1]{fontenc} 

\usepackage{tikz}
\usetikzlibrary{shapes}
\RequirePackage{ifthen}
\RequirePackage{amsmath}
\RequirePackage{amssymb}
\RequirePackage{xspace}
\RequirePackage{graphics}
\usepackage{color}
\usepackage{keyval}
\usepackage{xspace}
\usepackage{mathrsfs}
\usepackage{enumitem}

\newcommand{\sayan}[1]{\textcolor{blue}{#1}}

\newcommand{\reals}{{\mathbb{R}}} 

\newcommand{\nnreals}{{\mathbb{R}^{\geq 0}}}
\newcommand{\dom}{\relax\ifmmode {\mathit{dom}} \else ${\sf dom}$\fi}

\newcommand{\Reach}{\mathit{Reach}} 
\newcommand{\init}{\mathit{init}} 
 
\newcommand{\coverstack}{\mathit{traversalstack}}

\newcommand{\unsafeset}{{\mathit{O}}} 

\newcommand{\node}{\mathit{node}}
\newcommand{\nodechildren}{\mathit{children}}
\newcommand{\reachset}{\mathit{Reach}}

\newcommand{\safe}{\mathit{safe}}
\newcommand{\unsafe}{\mathit{unsafe}}
\newcommand{\unknown}{\mathit{unknown}}
\newcommand{\refine}{\mathit{refine}}
\newcommand{\state}{{x}}

\newcommand{\ourtool}{\sf{SceneChecker}}
\newcommand{\ourtacastool}{\sf{CacheReach}}

\newcommand{\initset}{\Theta}
\newcommand{\initpar}{s_{\mathit{init}}}
\newcommand{\finitpar}{p_{\mathit{init}}}
\newcommand{\initparv}{p_{\mathit{init,v}}}
\newcommand{\finitparv}{p_{\mathit{init,v}}}
\newcommand{\edgeset}{E}

\newcommand{\edgeinst}{e}

\newcommand{\guard}{\mathit{guard}}
\newcommand{\reset}{\mathit{reset}}

\newcommand{\stateset}{X}
\newcommand{\inputset}{{\bf U}}
\newcommand{\parset}{S}
\newcommand{\fparset}{P}
\newcommand{\goalset}{{\bf G}}
\newcommand{\ha}{{H}}
\newcommand{\hav}{A_v}
\newcommand{\haref}{H_v'}

\newcommand{\parinst}{s}
\newcommand{\fparinst}{p}
\newcommand{\fparinstvf}{p_{v,\mathit{1}}}
\newcommand{\fparinstvs}{p_{v,\mathit{2}}}
\newcommand{\parinstvf}{s_{v,\mathit{1}}}
\newcommand{\parinstvs}{s_{v,\mathit{2}}}

\newcommand{\stateinst}{x}
\newcommand{\inputinst}{u}
\newcommand{\lstate}{\mathit{lstate}}
\newcommand{\fstate}{\mathit{fstate}}

\newcommand{\roads}{{\bf R}}
\newcommand{\roadinst}{{r}}
\newcommand{\src}{\mathit{src}}
\newcommand{\dest}{\mathit{dest}}

\newcommand{\initmode}{s_{\mathit{init}}}
\newcommand{\finitmode}{p_{\mathit{init}}}

\newcommand{\initsetv}{X_{\mathit{init,v}}}
\newcommand{\initmodev}{s_{\mathit{init,v}}}

\newcommand{\fsr}{\mathcal{R}_{\mathit{rv}}}
\newcommand{\fsrf}{\mathcal{R}_{\mathit{1}}}
\newcommand{\fsrs}{\mathcal{R}_{\mathit{2}}}
\newcommand{\rv}{\mathit{rv}}
\newcommand{\newrv}{\mathit{rv}'}

\newcommand{\permodedict}{\mathit{Cache}}

\newcommand{\computereachset}{{\sf computeReachset}}

\newcommand{\splitmode}{{\sf splitMode}}

\newcommand{\cache}{\mathit{cache}}

\newcommand{\True}{\mathit{True}}

\newcommand{\pop}{\mathit{pop}}

\newcommand{\verify}{{\sf verify}}
\newcommand{\result}{\mathit{result}}

\newcommand{\fsrref}{\mathcal{R}_{\mathit{rv}'}}
\newcommand{\fixedpoint}{\mathit{fixedpoint}}
\newcommand{\issafe}{\mathit{issafe}}

\newcommand{\plan}{{G}}
\newcommand{\vertices}{V}

\newcommand{\segments}{S}

\newcommand{\timebound}{\mathit{tbound}}
\newcommand{\transpolyvirtual}{{\sf polyVir}}
\newcommand{\transvirtualpoly}{{{\sf virPoly}}}
\newcommand{\transmodevirtual}{{\sf modeVir}}
\newcommand{\abstractfunc}{{\sf abstract}}

\newcommand{\subtree}{{\sf subTree}}
\newcommand{\restartmode}{\mathit{restartp}}
\newcommand{\currinitset}{\mathit{initset}}
\newcommand{\agent}{\mathcal{A}}

\usepackage{setspace}
\usepackage{listings}

\usepackage[normalem]{ulem}

\usepackage{placeins}
\makeatletter
\newcommand\nocaption{%
    \renewcommand\p@subfigure{}
    \renewcommand\thesubfigure{\thefigure\alph{subfigure}}
}
\makeatother

\usepackage{tikz}
\usetikzlibrary{fit}
\usepackage{twoopt}

\newsavebox\IBoxA \newsavebox\IBoxB \newlength\IHeight
\newcommand\TwoFig[6]{
	\sbox\IBoxA{\includegraphics[width=0.45\textwidth]{#1}}
	\sbox\IBoxB{\includegraphics[width=0.45\textwidth]{#4}}%
	\ifdim\ht\IBoxA>\ht\IBoxB
	\setlength\IHeight{\ht\IBoxB}%
	\else\setlength\IHeight{\ht\IBoxA}\fi
	\begin{figure}[!htb]
		\minipage[t]{0.45\textwidth}\centering
		\includegraphics[height=\IHeight]{#1}
		\caption{#2}\label{#3}
		\endminipage\hfill
		\minipage[t]{0.45\textwidth}\centering
		\includegraphics[height=\IHeight]{#4}
		\caption{#5}\label{#6}
		\endminipage 
	\end{figure}%
}

\newcommandtwoopt\Textbox[5][2.5cm][2cm]{%
	\begin{tikzpicture}[remember picture,overlay]
	\coordinate (aux) at ([xshift=#1]#4);
	\node[inner ysep=3pt,yshift=0.6ex,draw=gray,thick,
	fit=(#3) (aux),baseline] 
	(box) {};
	\node[text width=#2,anchor=north east,
	font=\sffamily\footnotesize,align=right] 
	at (box.north east) {#5};
	\end{tikzpicture}%
}

\begin{document}
\title{$\ourtool$: Boosting Scenario Verification using Symmetry Abstractions}
%
%
\author{Hussein Sibai \and Yangge Li \and Sayan Mitra}
\institute{ University of Illinois at Urbana-Champaign \\ Coordinated Science Laboratory\\
\email{\{sibai2,li213,mitras\}@illinois.edu}}
%
%
%
\maketitle              
\begin{abstract}
We present $\ourtool$, a tool for verifying scenarios involving vehicles executing complex plans in large cluttered workspaces. 
$\ourtool$ converts the scenario verification problem to a standard hybrid system verification problem, and solves it effectively by exploiting structural properties in the plan and the vehicle dynamics.  $\ourtool$ uses symmetry abstractions,  a novel refinement algorithm, and importantly, is built to boost the performance of any existing reachability analysis tool as a plug-in subroutine.
We evaluated $\ourtool$ on several scenarios involving ground and aerial vehicles with nonlinear dynamics and neural network controllers, employing different kinds of symmetries, using different reachability subroutines, and following plans with hundreds of waypoints in complex workspaces. 
Compared to two leading tools, DryVR and Flow*, $\ourtool$ shows  20$\times$ speedup in verification time, even while using those very tools as reachability  subroutines.

\end{abstract}
\section{Introduction}
\label{sec:intro} 
Remarkable progress has been made in safety verification of hybrid and cyber-physical systems in the last decade~\cite{Spaceex,bak2017hylaa,flow,DuggiralaFM015,c2e2,DryVR,10.1145/3302504.3313351,NNV,ivanov2019verisig}. The methods and tools developed have been applied to check safety of aerospace, medical, and autonomous vehicle control systems~\cite{flow,DuggiralaFM015,Althoff2015a,NimaHSCC,FanQM0D16}. The next  barrier in making these techniques usable for more complex applications is to deal with what is colloquially called the {\em scenario verification problem}.
A key part of the scenario verification problem is to check that  a vehicle or an agent can execute a plan through a complex environment. 
A  planning algorithm (e.g., probabilistic roadmaps~\cite{PRM}  and RRT~\cite{Lavalle98rapidly-exploringrandom}) generates a set of possible paths avoiding obstacles, but only considering the geometry of the scenario, not the dynamics. The  verification task has to ensure that the plan can indeed be safely executed by the vehicle with all the dynamic constraints and the state estimation uncertainties. 
Indeed, one can view a scenario as a hybrid automaton  with the modes defined by the segments of the planner, but this  leads to massive models. Encoding such  automata in existing tools presents some practical hurdles. More importantly, analyzing such models is challenging as the  over-approximation errors and the analysis times  grow rapidly with  the number of transitions. At the same time, such large hybrid  verification problems also have lots of repetitions and symmetries, which suggest  new opportunities. 

We present $\ourtool$, a tool that implements a symmetry abstraction-refinement algorithm for efficient scenario verification. 
Symmetry abstractions significantly reduce the number of modes and edges of an automaton $\ha$ by grouping all modes that share symmetric continuous dynamics~\cite{sibai-tac-2020}. 
$\ourtool$ implements a novel refinement algorithm for symmetry abstractions and is able to use any existing reachability analysis tool  as a subroutine.
Our current implementation comes with plug-ins for using Flow*~\cite{flow} and DryVR~\cite{DryVR}.
$\ourtool$'s verification algorithm is sound, i.e., if it returns  $\safe$, then the reachset of $\ha$ indeed does not intersect the unsafe set. The algorithm is  lossless in the sense that if one can  prove safety without using  abstraction, then $\ourtool$ can also prove safety via abstraction-refinement, and typically a lot faster.

$\ourtool$ offers an easy interface to specify plans, agent dynamics, obstacles, initial uncertainty, and symmetry maps.
$\ourtool$ checks if a fixed point has been reached after each call to the reachability subroutine, avoiding repeating computations.  
First, $\ourtool$ represents the input scenario as a hybrid automaton $\ha$ where modes are defined by the plan's segments. It uses the symmetry maps provided by the user to construct an abstract automaton $\ha_v$.
$\ha_v$ represents another scenario with fewer segments, each representing a group of symmetric segments in $\ha$. 
A side effect of the abstraction is that
upon reaching waypoints in $\ha_v$, the agent's state resets non-deterministically to a set of possible states. For example, in the case of rotation and translation invariance, the abstract scenario would have a single segment for any group of segments with a unique length in the original scenario. $\ourtool$ refines $\ha_v$ by splitting one of its modes to two modes. That corresponds to representing a  group of symmetric segments with one more segment in the abstract scenario, capturing more accurately the original scenario\footnote{ A figure showing the architecture of $\ourtool$ can be found in Appendix~\ref{sec:toolarch}.}. 

We evaluated $\ourtool$ on several scenarios where car and quadrotor agents with nonlinear dynamics follow plans to reach several destinations in 2D and 3D workspaces with hundreds of waypoints and polytopic obstacles.
We considered different symmetries (translation and rotation invariance) and controllers (Proportional-Derivative (PD) and Neural Networks (NN)).
We compared the verification time of $\ourtool$ with DryVR and Flow* as reachability subroutines against  Flow* and DryVR
as standalone tools. $\ourtool$ is faster than both tools in all scenarios considered, achieving up to 20$\times$ speedup in verification time  (Table~\ref{tab:other_tool}). In certain scenarios where Flow* timed out (executing for more than 120 minutes), $\ourtool$ is able to complete verification in as fast as 12 minutes using Flow* as a subroutine. $\ourtool$ when using abstraction-refinement achieved 13$\times$ speedup in verification time over not using abstraction-refinement in scenarios with the NN-controlled quadrotor (Section~\ref{sec:ex}).




\paragraph{Related work}
The idea of using symmetries to accelerate verification  has been exploited in a number of contexts such as probabilistic models~\cite{KwiatkowskaNP06,AntunaACE15}, automata~\cite{EmersonS93,ClarkeJ95}, distributed architectures~\cite{parameterizedsynthesis}, and hardware~\cite{Makai_Barett_hardwaresymmetry,Pandey-Bryant-symmetry-hardware-1999,FPGAsymmetry2008}. Some symmetry utilization algorithms are implemented in Mur$\phi$~\cite{IpDill93} and  Uppaal~\cite{Hendriks04addingsymmetry}.

In our context of cyber-physical systems, 
Bak et al.~\cite{BakSym2015} suggested using symmetry maps, called {\em reachability reduction transformations}, to transform reachsets to symmetric reachsets for continuous dynamical systems modeling non-interacting vehicles. 
Maidens et al. \cite{MaidensSym} proposed a symmetry-based dimensionality reduction method for backward reachable set computations for discrete dynamical systems. Majumdar et al. \cite{Majumdar2017} proposed a safe motion planning algorithm that computes a family of reachsets offline and composes them online using symmetry. Bujorianu et al.~\cite{Bujorianu_Katoen_2008} presented a symmetry-based theory to reduce stochastic hybrid systems for faster reachability analysis and discussed the challenges of designing symmetry reduction techniques across mode transitions. 

In a more closely related research, Sibai et al. \cite{sibai-atva-2019} presented a modified version of DryVR that utilizes symmetry to cache reachsets aiming to accelerate simulation-based safety verification of continuous dynamical systems. 
The related tool $\ourtacastool$  implements a hybrid system verification algorithm that uses symmetry to accelerate reachability analysis \cite{sibai-tacas-2020}. 
$\ourtacastool$ caches and shares computed reachsets between different modes of non-interacting agents using symmetry. 
$\ourtool$ is based on the  theory of symmetry abstractions of hybrid automata presented in~\cite{sibai-tac-2020}. 
They suggested computing the reachset of the abstract automaton instead of the concrete one then transform it to the concrete reachset using symmetry maps to accelerate verification.
$\ourtool$ is built based on this line of work with significant algorithmic and engineering  improvements. 
In addition to the abstraction of \cite{sibai-tac-2020}, $\ourtool$ 1) maps the unsafe set to an abstract unsafe set and verifies the abstract automaton instead of the concrete one
and 
2) decreases the over-approximation error of the abstraction through refinement.
$\ourtool$ 
does not cache reachsets
and thus saves cache-access and reachset-transformation times and does not incur over-approximation errors due to caching that $\ourtacastool$ suffers from \cite{sibai-tacas-2020}. 
At the implementation level, $\ourtool$ accepts plans that are general directed graphs and polytopic unsafe sets while $\ourtacastool$ accepts only single-path plans and hyperrectangle unsafe sets. We show more than $30\times$ speedup in verification time while having more accurate verification results when comparing $\ourtool$ against $\ourtacastool$ (Table~\ref{tab:other_tool} in Section~\ref{sec:ex}).


\section{Specifying Scenarios in $\ourtool$}
\label{sec:problem}


A scenario verification problem is specified by a set of obstacles, a plan, and an agent that is supposed to execute the plan without running into the obstacles. 
For ground and air vehicles, for example, the agent moves in a subset of the 2D or the 3D Euclidean space called the {\em workspace\/}. 
 A {\em plan\/} is a directed graph $\plan = \langle \vertices, \segments \rangle$ with vertices $\vertices$ in the workspace  called {\em waypoints\/} and edges $\segments$ called {\em segments}\footnote{\scriptsize We introduce this redundant nomenclature because later we will reserve the term edges to talk about mode transitions in hybrid automata. We use waypoints instead of vertices as a more natural term for points that vehicles have to follow.}.
A general graph allows for nondeterministic and contingency planning. 

%

An {\em agent} is a control system that can follow waypoints. Let the state space of the agent be $\stateset$ and  $\Theta \subseteq \stateset$ be the uncertain initial set. 
Let $\initpar$ be the initial segment in  $\plan$ that the agent has to follow.
From any state $x\in \stateset$, the agent follows a segment $s\in S$ by moving along a {\em trajectory}. 
A trajectory  is 
 a  function  $\xi:X\times S \times \nnreals \rightarrow X$ that meets certain dynamical constraints of the vehicle. 
Dynamics are either specified by ordinary differential equations (ODE) or by a black-box simulator.  For ODE models, $\xi$ is a solution of an equation of the form:
$\frac{d \xi}{dt}(\stateinst, \parinst, t) = f(\xi(\stateinst, \parinst, t),\parinst)$,
for any $t \in \nnreals$ and $\xi(\stateinst, \parinst, 0) = \stateinst$, where $f: \stateset \times \parset \rightarrow \stateset$ is Lipschitz continuous in the first argument. 
Note that the trajectories  only depend on the segment the agent is following (and  not on the full plan $\plan$). We denote by   $\xi.\fstate$, $\xi.\lstate$, and $\xi.\dom$ the initial and last states and the time domain of the time bounded trajectory $\xi$, respectively.

We can view the obstacles near each segment as sets of unsafe states,  $\unsafeset: \parset \rightarrow 2^\stateset$.
The map $\timebound: \parset \rightarrow \nnreals$ determines the maximum time the agent should spend in following any segment. 
For any pair of consecutive segments $(\parinst,\parinst')$, i.e. sharing a common waypoint in $\plan$, $\guard((\parinst,\parinst'))$ defines the set of states (a hyperrectangle around a waypoint) at which the agent is allowed 
to transition from following $\parinst$ to following $\parinst'$.


\begin{tcolorbox}
\paragraph{${\sf Scenario}$ ${\sf JSON}$ ${\sf file}$} is the first of the two user inputs. It specifies the scenario: $\Theta$ as a hyperrectangle; $S$ as a list of lists each representing two waypoints; $\guard$ as a list of hyperrectangles; $\timebound$ as a list of floats; and $\unsafeset$ as a list of polytopes. 
\\
\paragraph{${\sf Output}$}of $\ourtool$ is the scenario verification result  ($\safe$ or $\unknown$) and a number of useful performance  metrics, such as the number of mode-splits, number of  reachability calls, reachsets computation time, and total  time. $\ourtool$ can also visualize the various computed reachsets.
\end{tcolorbox}

\section{Transforming Scenarios to Hybrid Automata}
\label{sec:scene_to_automaton}
The input scenario is first represented as a hybrid automaton by a ${\sf Hybrid}$ ${\sf constructor}$. This constructor is a Python function that parses the ${\sf Scenario}$ ${\sf file}$  and constructs the data structures to store the scenario's hybrid automaton components. In what follows, we describe the constructed automaton informally. In our current implementation, sets are represented either as hyper-rectangles or as polytopes using the Tulip Polytope Library\footnote{\scriptsize \url{https://pypi.org/project/polytope/}}.

\paragraph{Scenario as a hybrid automaton} 
A hybrid automaton has a set of {\em modes\/} (or discrete states) and a set of continuous states. The evolution of the continuous states in each mode is specified by a set of trajectories and the transition across the modes are specified by $\guard$ and $\reset$ maps. 
The agent following a plan in a workspace can be naturally modeled as a hybrid automaton $\ha$,
where $\parinst_\init$ and $\Theta$ are 
its initial mode and set of states.

Each segment $\parinst \in \parset$ of the plan $\plan$ defines a {\em mode} of $\ha$. 
The set of edges $\edgeset \subseteq \parset \times \parset$ of $\ha$ is defined as pairs of consecutive segments in $\plan$.
For an edge $\edgeinst \in \edgeset$, $\guard(\edgeinst)$ is the same as that of $\plan$. 
The $\reset$ map of $\ha$ is the identity map.
%
We will see in Section~\ref{sec:symmetry_abstractions} that abstract automata will have nontrivial reset maps.

\paragraph{Verification problem}
An {\em execution} of length $k$ is a sequence $\sigma := (\xi_0, \parinst_0), \ldots, (\xi_k,\parinst_k)$. It models the behavior of the agent following a particular path in the plan $G$. An execution $\sigma$ must satisfy:
1) $\xi_0.\fstate \in \initset$ and $\parinst_0 =\parinst_\init$,
for each $i\in \{0,\ldots, k-1\}$,
2) $(\parinst_i,\parinst_{i+1}) \in \edgeset$, 
3) $\xi_i.\lstate \in \guard((\parinst_i, \parinst_{i+1}))$, and 
4) $\xi_i.\lstate = \xi_{i+1}.\fstate$,
and 
5) for each $i\in \{0,\ldots, k\}$, $\xi_i.\dom \leq \timebound(\parinst_i).$
%
The set of {\em reachable states\/}  is  $\reachset_\ha := \{ \sigma.\lstate \ |\ \sigma$ is an execution$\}$.
The restriction of $\reachset_\ha$ to states with mode $\parinst \in \parset$ (i.e., agent following segment $\parinst$) is  denoted by $\reachset_\ha(\parinst)$.
Thus, the hybrid system verification problem requires us to check 
whether $\forall \parinst \in \parset$,  
$\reachset_\ha(\parinst) \cap \unsafeset(\parinst) = \emptyset$.




\section{Specifying Symmetry Maps in $\ourtool$}
\label{sec:symdef}

The hybrid automaton representing a scenario, as constructed by the ${\sf Hybrid}$ ${\sf constructor}$, is transformed into an abstract automaton. 
%
$\ourtool$ uses symmetry abstractions~\cite{sibai-tac-2020}. 
The abstraction is constructed by the $\abstractfunc$ function (line~\ref{ln:abstract} of Algorithm~\ref{code:bettersafetyVerifAlgo}) which uses a collection of pairs of maps  $\Phi = \{(\gamma_{\parinst}: \stateset \rightarrow \stateset,\rho_{\parinst}: \parset \rightarrow \parset)\}_{\parinst\in \parset}$ that is provided by the user.
We describe below how these maps are specified by the user in the ${\sf Dynamics}$ ${\sf file}$.
%
These maps should satisfy:
\begin{align}
\label{eq:detailed_transformation}
\forall\ t\geq 0, \stateinst_0 \in \stateset, \parinst \in \parset, \gamma_\parinst(\xi(\stateinst_0, \parinst,t)) = \xi(\gamma_\parinst(\stateinst_0), \rho_\parinst(\parinst), t).  
\end{align} 
where $\forall \parinst \in \parset$, the map $\gamma_\parinst$ is differentiable and invertible. 
Such maps are  called {\em symmetries} for the agent's dynamics. They transform the agent's trajectories to other symmetric ones of its trajectories starting from symmetric initial states and following symmetric modes (or segments in our scenario verification setting).
It is worth noting that~(\ref{eq:detailed_transformation}) does not depend on whether the trajectories $\xi$ are defined by ODEs or black-box simulators. 
Currently, condition~(\ref{eq:detailed_transformation}) is not checked by $\ourtool$ for the maps specified by the user. However, in the following discussion, we present some ways for the user to check~(\ref{eq:detailed_transformation}) on their own.
For ODE models, a sufficient condition for (\ref{eq:detailed_transformation}) to be satisfied is if:
$\forall\ \stateinst \in \stateset,\parinst \in \parset,\ \frac{\partial \gamma_\parinst}{\partial \stateinst} f(\stateinst, \parinst) = f(\gamma_\parinst(\stateinst), \rho_\parinst(\parinst))$, where $f$ is the right-hand-side of the ODE~\cite{russo2011symmetries}. For black-box models, $(\ref{eq:detailed_transformation})$ can be checked using sampling methods. In realistic settings, dynamics might not be exactly symmetric due to unmodeled uncertainties. In the future, we plan to account for such uncertainties as part of the reachability analysis.

In scenario verification, a given workspace would have a coordinate system according to which the plan (waypoints)  and the agent's state (position, velocity, heading angle, etc.) are represented. In a 2D workspace, for any segment $\parinst \in \parset$, an example symmetry $\rho_\parinst$ would transform the two waypoints of $s$ to a new coordinate system where the second waypoint is the origin and $\parinst$ is aligned with the negative side of the horizontal axis. 
%
The corresponding $\gamma_\parinst$ would transform the agent's state to this new coordinate system (e.g. by rotating its position and velocity vectors and shifting the heading angle). For such a pair $(\gamma_\parinst,\rho_\parinst)$ to satisfy (\ref{eq:detailed_transformation}), the agent's dynamics have to be invariant to such a coordinate transformation and (\ref{eq:detailed_transformation}) merely formalizes this requirement. 
Such an invariance property is expected from vehicles' dynamics--rotating or translating the lane should not change how an autonomous car behaves.


\begin{tcolorbox}
\paragraph{${\sf Dynamics}$ ${\sf file}$\/} is the second input provided by the user in addition to the ${\sf Scenario}$  ${\sf file}$  and it contains the following:


\begin{description}
    \item$\transpolyvirtual(\stateset',\parinst)$: returns $\gamma_\parinst(\stateset')$ for any polytope $\stateset' \subset \stateset$ and segment $\parinst \in \parset$.
    
    \item $\transmodevirtual(\parinst)$: returns $\rho_\parinst(\parinst)$ for any given segment $\parinst \in \parset$.
     
    \item $\transvirtualpoly(\stateset',\parinst)$:  returns $\gamma_\parinst^{-1}(\stateset')$, implementing the inverse of $\transpolyvirtual$. 
    \item 
    $\computereachset(\currinitset,\parinst,T)$: returns
    a list of hyperrectangles over-approximating the agent's reachset starting from $\currinitset$ following segment $\parinst$ for $T$ time units, for any set of states $\currinitset \subset \stateset$, segment $\parinst \in \parset$, and $T \geq 0$.
\end{description}
\end{tcolorbox}

\section{Symmetry Abstraction of the Scenario's Automaton}
\label{sec:symmetry_abstractions}
In this section, we describe how the  $\abstractfunc$ function in Algorithm~\ref{code:bettersafetyVerifAlgo} uses the functions in the ${\sf Dynamics}$ ${\sf file}$ to construct an abstraction of the scenario's hybrid automaton provided by the ${\sf Hybrid}$ ${\sf constructor}$.
Given the symmetry maps of $\Phi$, the symmetry abstraction of $\ha$ is another hybrid automaton $\ha_v$
that aggregates many symmetric modes (segments) of $\ha$ into a single mode of $\ha_v$.

\paragraph{Modes and Transitions}
Any segment $\parinst \in \parset$ of $\ha$ is mapped to the segment $\rho_\parinst(\parinst)$ in $\ha_v$ using $\transmodevirtual$. The set of modes $\parset_v$ of $\ha_v$ is the set of segments $\{\rho_\parinst(\parinst)\}_{\parinst \in \parset}$. For any $\parinst_v$,  $\timebound_v(\parinst_v) = \max_{\parinst \in \parset, \parinst_v = \rho_\parinst(\parinst)} \timebound(\parinst)$.
In the example of Section~\ref{sec:symdef}, the segments in $\ha_v$ are aligned with the horizontal axis and ending at the origin. The number of segments in $\ha_v$ would be the number of segments in $\plan$ with unique lengths. The agent would always be moving towards the origin of the workspace in the abstract scenario.
Any edge $\edgeinst = (\parinst, \parinst') \in \edgeset$ of $\ha$ is mapped to the edge $\edgeinst_v = (\rho_\parinst(\parinst), \rho_{\parinst'}(\parinst'))$ in $\ha_v$.
The $\guard(\edgeinst)$ is mapped to $\gamma_\parinst(\guard(\edgeinst))$  using $\transpolyvirtual$ which becomes part of $\guard_v(\edgeinst_v)$ in $\ha_v$. For any $\stateinst \in \stateset$, $\reset(\stateinst, \edgeinst)$, which is equal to $\stateinst$, is mapped to $\gamma_{\parinst'} (\gamma_\parinst^{-1}(\stateinst))$ and becomes part of $\reset_v(\stateinst, \edgeinst_v)$ in $\ha_v$. 
In our example in Section~\ref{sec:symdef}, 
the $\gamma_\parinst^{-1}(\stateinst)$  would represent $\stateinst$ in the absolute coordinate system assuming it was represented in the coordinate system defined by segment $\parinst$. The $\gamma_{\parinst'}(\gamma_\parinst^{-1}(\stateinst))$ would represent $\gamma_\parinst^{-1}(\stateinst)$ in the new coordinate system defined by segment $\parinst'$. 
The $\guard_v(\edgeinst_v)$ would be the union of rotated hyperrectangles centered at the origin that result from translating and rotating the guards of the edges represented by $\edgeinst_v$. 
%
 The initial set $\Theta$ of $\ha$ is mapped to $\Theta_v = \gamma_{\initpar}(\Theta)$, the initial set of $\ha_v$.
A formal definition of symmetry abstractions can be found in Appendix~\ref{sec:abstraction_definition} (or \cite{sibai-tac-2020}). 

The unsafe map  $\unsafeset$ is mapped to $\unsafeset_v$, where $\forall \parinst_v \in \parset_v, \unsafeset_v(\parinst_v) = \cup_{\parinst \in \parset, \rho_\parinst(\parinst) = \parinst_v} \gamma_\parinst(\unsafeset(\parinst))$. That means the obstacles near any segment $\parinst \in \parset$ in the environment will be mapped to be near its representative segment $\rho_\parinst(\parinst)$ in $\ha_v$.

A forward simulation relation between $\ha$ and $\ha_v$ can show that if $\ha_v$ is safe with respect to $\unsafeset_v$, then $\ha$ is safe with respect to $\unsafeset$. More formally, if $\forall \parinst_v \in \parset_v, \Reach_{\ha_v}(\parinst_v) \cap \unsafeset_v(\parinst_v) = \emptyset$, then $\forall \parinst \in \parset, \Reach_\ha(\parinst) \cap \unsafeset(\parinst) = \emptyset$ \cite{sibai-tac-2020}.

\section{$\ourtool$ Algorithm Overview}

A sketch of the core abstraction-refinement algorithm is shown in Algorithm~\ref{code:bettersafetyVerifAlgo}. 
It constructs a symmetry abstraction $\ha_v$ of the concrete automaton $\ha$ resulting from the ${\sf Hybrid}$ ${\sf constructor}$. 
$\ourtool$ attempts to verify the safety of $\ha_v$ using traditional reachability analysis. 
$\ourtool$ uses a $\cache$  to store
per-mode initial sets from which reachsets have been computed
and thus avoids repeating  computations.

\begin{algorithm}
	\caption{${\ourtool}(\Phi = \{(\gamma_\parinst, \rho_\parinst)\}_{\parinst \in \parset}, \ha, \unsafeset)$}
	\label{code:bettersafetyVerifAlgo}
	\begin{algorithmic}[1]
		\State $\ha_v, \unsafeset_v \gets \abstractfunc(\ha, \unsafeset, \Phi)$ \label{ln:abstract}
		\State $\forall \parinst \in \parset, \rv[\parinst] \gets \rho_\parinst(\parinst)$ \label{ln:rv_definition}
		\While{$\True$} \label{ln:whileloop}
		\State $\cache \gets \{\parinst_v \mapsto \emptyset\ |\ \parinst_v \in \parset_v  \}$ \label{ln:cache_init}
		\State $\result, \parinst_v^* \gets \verify(\rv[\initpar], \initset_v, \cache, \rv, \ha_v,\unsafeset_v)$ \label{ln:main_symar_call}
		\If{$\result = \safe$ or $\unknown$} {\bf return:} $\result$ \label{ln:return_result}
		\Else$\ \rv, \ha_v, \unsafeset_v \gets \splitmode(\parinst_v^*, \rv, \ha_v, \unsafeset_v, \ha, \unsafeset)$ \label{ln:refine}
		\EndIf
		\EndWhile
		\end{algorithmic}
\end{algorithm} 

The core algorithm $\verify$ (Algorithm~\ref{code:symar}) is called iteratively. If $\verify$ returns $(\safe, \bot)$ or $(\unknown,\bot)$, $\ourtool$ returns the same result. If $\verify$ instead results in $(\refine, \parinst_v^*)$, $\splitmode$ (check Appendix~\ref{sec:mode_split_algorithm} for the formal definition) is called to refine $\ha_v$ by splitting $\parinst_v^*$ into two modes $\parinst_v^1$ and $\parinst_v^2$.
Each of the two modes would represent  part of the set of the segments of $\parset$ that were originally mapped to $\parinst_v$ in $\rv$. Then the edges, guards, resets, and the unsafe sets related to $\parinst_v$  are split according to their definitions.
%


The function $\verify$ executes a {\em depth first search} (DFS) over the mode graph of $\ha_v$. For any mode $\parinst_v$ being visited,  $\computereachset$ computes  $R_v$, an over-approximation of the agent's reachset starting from $\currinitset$ following segment $\parinst_v$ for time $\timebound_v(\parinst_v)$. 
If $R_v \cap \unsafeset_v(\parinst_v) = \emptyset$, $\verify$ recursively calls $\parinst_v$'s children continuing the DFS in line~\ref{ln:call_symar_for_children}. Before calling each child, its initial set is computed and the part for which a reachset has already been  computed and stored in $\cache$ is subtracted. 
If all calls return $\safe$, then 
$\currinitset$
is added to the other initial sets in $\cache[\parinst_v]$ (line~\ref{ln:storereachset}) and $\verify$ returns $\safe$. 
Most importantly, if $\verify$ returns $(\refine,\parinst_v^*)$ for any of $\parinst_v$'s children, it directly returns $(\refine,\parinst_v^*)$ for $\parinst_v$ as well (line~\ref{ln:check_if_children_refine}). If any child returns $\unknown$ or $R_v$ intersects $\unsafeset_v(\parinst_v)$, $\verify$ will need to split $\parinst_v$. In that case, it checks 
if $\rv^{-1}[\parinst_v]$ is not a singleton set and thus amenable to splitting (line~\ref{ln:check_if_can_be_split}). If $\parinst_v$ can be split, $\verify$ returns $(\refine, \parinst_v)$.
Otherwise, $\verify$ returns $(\unknown, \bot)$ implicitly asking one of $\parinst_v$'s ancestors to be split instead.

\paragraph{Correctness} $\ourtool$ ensures that all the refined automata $\ha_v$'s are abstractions of the original hybrid automaton $\ha$
(a proof is given in Appendix~\ref{sec:mode_split_algorithm}). 
For any mode with a reachset intersecting the unsafe set, $\ourtool$ keeps refining that mode and its ancestors until safety can be proven or $\ha_v$ becomes $\ha$. 

\begin{theorem}[Soundness]
\label{thm:soundness}
If $\ourtool$ returns $\safe$, then $\ha$ is $\safe$.
\end{theorem}

\begin{algorithm}
	\caption{$\verify(\parinst_v, \currinitset, \cache, \rv, \ha_v, \unsafeset_v)$}
	\label{code:symar}
	\begin{algorithmic}[1]
		\State $R_v \gets \computereachset(\currinitset, \parinst_v)$ \label{ln:computereachset}
		\If {$R_v \cap \unsafeset_v(\parinst_v) = \emptyset$}
		\label{ln:checksafety}
		\For{$ \parinst_v' \in \nodechildren(\parinst_v)$} \label{ln:iterate_over_children}
		\State $\currinitset' \gets \reset_v(\guard_v((\parinst_v, \parinst_v')) \cap R_v) \textbackslash \cache[\parinst_v']$  
		\label{ln:initset}
		\If{$\currinitset' \neq \emptyset$} 
		\State $\result, \parinst_v^* \gets \verify(\parinst_v', \currinitset', \cache, \rv, \ha_v,\unsafeset_v)$ \label{ln:call_symar_for_children}
		\If{$\result = \refine$} {\bf return:} $\refine, \parinst_v^*$ \label{ln:check_if_children_refine}
		\ElsIf{$\result = \unknown$} {\bf break} \label{ln:check_if_children_unknown}
		\EndIf
		\EndIf \label{ln:reached_fixed_point}
		\EndFor
		\EndIf
		\If{$R_v \cap \unsafeset_v(\parinst_v) \neq \emptyset$ or $\result$ is $\unknown$} \label{ln:check_if_unsafe_or_children_unknown}
		\If{$|\rv^{-1}[\parinst_v]| > 1$} 
		\textbf{return:} $\refine, \parinst_v$ \label{ln:check_if_can_be_split}
		\Else \textbf{ return:} $\unknown, \bot$ \label{ln:else_let_parent_refine}
		\EndIf
		\EndIf
		\State $\cache[\parinst_v] \gets \cache[\parinst_v] \cup \currinitset$ 
		\label{ln:storereachset}
		\State \textbf{return:} $\safe, \bot$ 
		\label{ver_ln:return_safe}
	\end{algorithmic}
\end{algorithm} 

\noindent
If $\verify$ is provided with the concrete automaton $\ha$ and unsafe set $\unsafeset$, it will be the traditional safety verification algorithm having no over-approximation error due to abstraction. If such a call to $\verify$ returns $\safe$, then $\ourtool$ is guaranteed to return $\safe$. That means that the refinement ensures that the over-approximation error of the reachset caused by the abstraction is reduced to not alter the verification result. 

%
\paragraph{Counter-examples}
$\ourtool$ currently does not find counter-examples to show that the scenario is $\unsafe$. There are several sources of over-approximation errors, namely, $\computereachset$ and guard intersections. Even after all the over-approximation errors from symmetry abstractions are eliminated, as refinement does, it still cannot infer unsafe executions or counter-examples because of the other errors.
We plan to address this in the future by combining the current algorithm with systematic simulations.

\section{Experimental Evaluation}
\label{sec:ex} 

\paragraph{Agents and controllers}
In our experiments, we consider two types of nonlinear agent models:  
a standard 3-dimensional car (C) with bicycle dynamics and $2$ inputs, and a 6-dimensional quadrotor (Q) with $3$ inputs.
For each of these agents, we developed a 
 PD controller and  a NN controller for tracking segments. 
The NN controller for the quadrotor is from Verisig's paper~\cite{ivanov2019verisig} (Appendix~\ref{sec:quadrotor-case-study} for  more details) but modified  to be rotation symmetric (Appendix~\ref{sec:nonsymcontrol} for more details). 
Similarly, the NN controller for the car is also rotation symmetric. Both NN controllers are translation symmetric as they take as input the difference between the agent's state and the segment being followed.
The PD controllers are translation and rotation symmetric by design. 

\paragraph{Symmetries}
We experimented with two different collections of symmetry maps $\Phi$s: 1) translation symmetry (T), where for any  segment $\parinst$ in $\plan$, $\gamma_\parinst$ maps the states so that the coordinate system is translated by a vector that makes its origin at the end waypoint of $\parinst$, and 2) rotation and translation symmetry (TR), where instead of just translating the origin, $\Phi$ rotates the $xy$-plane so that $\parinst$ is aligned with the $x$-axis, which we described in Section~\ref{sec:symdef}.
For each agent and one of its controllers, we manually verified that condition (\ref{eq:detailed_transformation}) is satisfied for each of the two $\Phi$s using the sufficient condition for ODEs in Section~\ref{sec:symdef}.


\paragraph{Scenarios}
We created four scenarios with 2D workspaces  (S1-4)  and one scenario with a 3D workspace (S5) with corresponding plans. We generated the plans using an RRT planner~\cite{10.1007/978-3-030-53288-8_31} after specifying a number of goal sets that should be reached. 
We modified S4 to have more obstacles but still have the same plan and named the new version S4.b and the original one S4.a.
When the quadrotor was considered, the waypoints of the 2D scenarios (S1-4) were converted to 3D representation by setting the altitude for each waypoint to 0. Scenario S5 is the same as S2 but S5's waypoints have varying altitudes.
The scenarios have different complexities ranging from few segments and obstacles to hundreds of them. All scenarios are safe when traversed by any of the two agents.
%

We verify these scenarios using two instances of $\ourtool$, one with DryVR and the other with Flow*, implementing $\computereachset$.
$\ourtool$ is able to verify  all scenarios with PD controllers. 
The results are shown in Table~\ref{tab:other_tool}\footnote{\scriptsize Figures presenting the reachsets of the concrete and abstract automata for different scenarios can be found in Appendix~\ref{sec:reachset_figures}. The machine specifications can be found in Appendix~\ref{sec:machine_specification}.}. 
\paragraph{Observation 1: $\ourtool$ offers fast scenario verification and boosts existing reachability tools}
Looking at the two total time (Tt) columns for the two instances of $\ourtool$ with the corresponding columns for Flow* and DryVR, it becomes clear that symmetry abstractions can boost the verification performance of reachability engines. 
For example, in C-S4.a, $\ourtool$ with DryVR was around $20\times$ faster than DryVR. In C-S3, $\ourtool$ with Flow* was around $16\times$ faster than Flow*. In scenario Q-S5, $\ourtool$ timed out at least in part because a $\computereachset$ call to Flow* timed out.
Even when many refinements are required and thus causing several repetitions of the verification process in Algorithm~\ref{code:bettersafetyVerifAlgo}, $\ourtool$ is still faster than DryVR and Flow* (C-S4.b). All three tools resulted in $\safe$ for all scenarios when completed executions.

\begin{table}[!htp]\centering
\vspace{-0.3in}
\caption{\scriptsize Comparison between $\ourtool$, DryVR (DR), Flow* (F*), and $\ourtacastool$ (${\sf CacheR}$). Both $\ourtool$ and $\ourtacastool$ use reachability tools as subroutines. The subroutines used are specified after the '+' sign.
$\Phi$ is TR. 
The table shows the number of mode-splits performed (Nrefs), the total number of calls to $\computereachset$ (Rc), the total time spent in reachset computations (Rt), and the total computation time in minutes (Tt). In  scenarios where a tool ran over 120 minutes, we marked the Tt column as `Timed out'(TO) and when it errored, we marked it as `Not Available'(NA).  }\label{tab:other_tool}
\small
\begin{tabular}{lr|rrrr|rr|r|rrrr|rr|rr}\toprule
& &\multicolumn{4}{c|}{$\ourtool$+DR} & \multicolumn{2}{c|}{${\sf CacheR}$+DR}  &DR &\multicolumn{4}{c|}{$\ourtool$+F*} & \multicolumn{2}{c|}{${\sf CacheR}$+F*} &F*\\\cmidrule{3-16}
Sc. &$|\parset|$ &Nrefs &Rc &Rt &Tt &Rc &Tt&Tt &NRefs &Rc &Rt &Tt&Rc &Tt &Tt  \\\midrule
C-S1 &6 &1 &4 &0.14 &0.15 &46 &1.75&1.34 &1 &4 &0.51 &0.52 &52 &8.20&2.11 \\
C-S2 &140 &0 &1 &0.04 &0.66 &453.86 &37.42&11.25 &0 &1 &0.18 &0.79 &192 & 30.95&17.52 \\
C-S3 &458 &0 &1 &0.04 &4.26 &398.26 &33.32&75.35 &0 &1 &0.11 &4.34 &176 & 28.64&73.06 \\
C-S4.a &520 &2 &7 &0.26 &4.52 & 276.64 &23.23&95.02  &2 &7 &0.80 &4.96 &160 & 25.98 &61.53 \\
C-S4.b &520 &10 &39 &1.48 &8.90 &277.10 &23.17&95.05 &10 &39 &2.83 &31.73 &160 & 26.07&60.67 \\
Q-S1 &6 &1 &4 &0.05 &0.06 &NA&NA&0.25 &1 &4 &13.85 &14.13 &NA &TO &30.17 \\
Q-S2 &140 &0 &1 &0.04 &0.88 & NA& NA&4.93&0 &1 &3.38&12.62& NA&TO &TO   \\
Q-S3 &458 &0 &1 &0.06 &5.9 & NA& NA&45.03&0 &1 &4.98 &62.66& NA&TO&TO   \\
Q-S4.a &520 &0 &1 &0.06 &3.32 & NA& NA&55.99&0 &1 &4.8 &34.89  & NA&TO&TO \\
Q-S5 &280 &0 &36 &0.85 &3.06 & NA& NA&4.91&NA  &NA  &NA  & TO & NA&TO &TO \\
\bottomrule
\end{tabular}
\vspace{-0.3in}
\end{table}

\paragraph{Observation 2: $\ourtool$ is faster and more accurate than  $\ourtacastool$} Since $\ourtacastool$ only handles single-path plans, we only verify the longest path in the plans of the scenarios in its experiments.
%
%
$\ourtacastool$'s instance with Flow* resulted in unsafe reachsets in C-S1 and C-S4.b scenarios likely because of the caching over-approximation error. In all scenarios where $\ourtacastool$ completed verification besides C-S4.b, it has more Rc and longer Tt (more than $30\times$ in C-S2) while verifying simpler plans than $\ourtool$ using the same reachability subroutine.
In all Q scenarios, $\ourtacastool$'s instance with Flow* timed out, while its instance with DryVR terminated with an error.  

\paragraph{Observation 3: More symmetric dynamics result in faster verification time} $\ourtool$ usually runs slower in 3D scenarios compared to 2D ones (Q-S2 vs. Q-S5) in part because there is no rotational symmetry in the $z$-dimension to exploit. That leads to larger abstract automata. Therefore, many more calls to $\computereachset$ are required.

%

We only used $\ourtool$'s instance with DryVR for agents with NN-controllers\footnote{\scriptsize Check Appendix \ref{sec:other_tools} for a discussion about our attempts for using other verification tools for NN-controlled systems as reachability subroutines.}.
 We tried different $\Phi$s. The results are shown in Table~\ref{tab:refine}. When not using abstraction-refinement, $\ourtool$ took 
11, 132, and 73
minutes for the QNN-S2, QNN-S3, and QNN-S4 scenarios, while DryVR took 5, 46, and 55 minutes for the same scenarios, respectively.
Comparing these results with those in Table~\ref{tab:refine} shows that the speedup in verification time of $\ourtool$ is caused by the abstraction-refinement algorithm, achieving more than 13$\times$ in certain scenarios (QNN-S4 using $\Phi=$ T). $\ourtool$'s instance with DryVR was more than 10$\times$ faster than DryVR in the same scenario.
\begin{table}[!htp]\centering
\vspace{-0.3in}
\caption{
\scriptsize Comparison between $\Phi$s.
In addition to the statisitics of Table~\ref{tab:other_tool}, this table reports the number of modes and edges in the initial and final (after refinement) abstractions ($|\parset_v|^i$, $|\edgeset_v|^i$; $|\parset_v|^f$, and $|\edgeset_v|^f$, respectively)
}\label{tab:refine}
\small
\begin{tabular}{l|rrrrrrrrrrrr}\toprule
Sc. &NRef &$\Phi$ &$|\parset|$&$|\parset_v|^i$ &$|\edgeset_v|^i$ &$|\parset_v|^f$ &$|\edgeset_v|^f$  &Rc &Rt &Tt \\\midrule
CNN-S2 &7 &TR &140 &1 &1 &8 &20 &35 &2.83 &5.64 \\
CNN-S4 &10 &TR &520 &1 &1 &11 &32 &68 &5.57 &36.66 \\
QNN-S2 &3 &TR &140 &1 &1 &4 &9 &9 &0.61 &4.01 \\
QNN-S3 &7 &TR &458  &1 &1 &8 &23 &21 &2.11 &13.98 \\
QNN-S4 &6 &TR &520 &1 &1 &7 &20  &15 &1.51 &8.11 \\
QNN-S2 &0 &T &140 &7 &19 &7 &19  &9 &0.62 &1.85 \\
QNN-S3 &4 &T &458 &7 &30 &11 &58 &29 &2.85 &16.72 \\
QNN-S4 &0 &T &520 &7 &30 &7 &30 &13 &1.3 &5.32 \\
\bottomrule
\end{tabular}
\vspace{-0.3in}
\end{table}


\paragraph{Observation 4: Choice of $\Phi$ is a trade-off between over-approximation error and number of refinements} 
The choice of $\Phi$ affects the number of refinements performed and the total running times (e.g. QNN-S2, QNN-S3, and QNN-S4). Using TR leads to a more succinct $\ha_v$ but larger over-approximation error causing more mode splits. On the other hand, using T leads to a larger $\ha_v$ but less over-approximation error and thus fewer refinements. This trade-off can be seen in Table~\ref{tab:refine}. For example, QNN-S4 with $\Phi= $T resulted in zero mode splits leading to $|\parset_v|^i = |\parset_v|^f = 7$, while $\Phi=$ TR resulted in 6 mode splits, starting with $|\parset_v|^i = 1$ modes and ending with  $|\parset_v|^f = 7$, and longer verification time because of refinements. On the other hand, in QNN-S3, $\Phi =$ TR resulted in Nref$=7$, $|\parset_v|^f = 8$, and Tt$=13.98$ min while $\Phi=$T resulted in Nref$=4$, $|\parset_v|^f = 11$, and Tt$=16.72$ min.
\paragraph{Observation 5: Complicated dynamics require more verification time} Different vehicle dynamics affect the number of refinements performed and consequently the verification time (e.g. QNN-S2, QNN-S4, CNN-S2, and CNN-S4). The car appears to be less stable than the quadrotor leading to longer verification time for the same scenarios. This can also be seen by comparing the results of Tables~\ref{tab:other_tool} and \ref{tab:refine}. The PD controllers lead to more stable dynamics than the NN controllers requiring less total computation time for both agents. More stable dynamics lead to tighter reachsets and fewer refinements. 

\section{Limitations and Discussions}
$\ourtool$ allows the choice of modes to be changed from segments to waypoints or sequences of segments as well. The waypoint-defined modes eliminate the need for segments of $\plan$ to have few unique lengths, but only allow $\Phi=$T. 
$\ourtool$ splits only one mode per refinement and then repeats the computation from scratch. It has to refine many times in unsafe scenarios until reaching the result $\unknown$.  
We plan to investigate other strategies for eliminating spurious counter-examples and returning valid ones in unsafe cases.  
In the future, it will be important to address other sources of uncertainty in scene verification such as moving obstacles, interactive agents, and other types of symmetries such as permutation and time scaling. 
Finally, it will be useful to connect a translator to generate scene files from common road simulation frameworks such as CARLA~\cite{Dosovitskiy17}, commonroad~\cite{commonroad}, and Scenic~\cite{scenic}.


\bibliographystyle{splncs03_unsrt}
\bibliography{samplepaper}

\appendix
\newpage

\section{Tool architecture}
\label{sec:toolarch}

\begin{figure}
  \begin{center}
    \includegraphics[width=\textwidth]{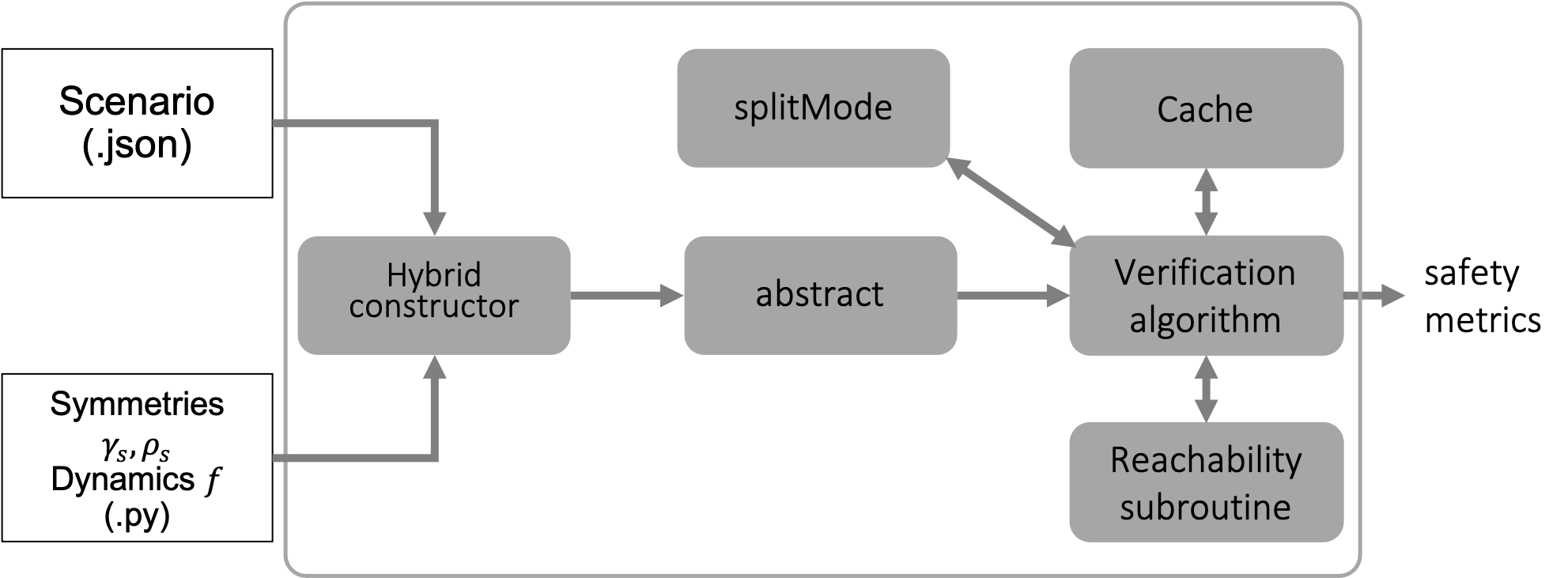}
  \end{center}
  \caption{$\ourtool$'s architecture. Its inputs are the ${\sf Scenario}$ ${\sf file}$ and the ${\sf Dynamics}$ ${\sf file}$. It outputs the safety results and other performance metrics. The ${\sf Hybrid}$ ${\sf constructor}$ component constructs a hybrid automaton from the given scenario. The $\abstractfunc$ component constructs the symmetry abstraction. The ${\sf Verification}$ ${\sf algorithm}$ component implements Algorithms~\ref{code:bettersafetyVerifAlgo} and \ref{code:symar}. The $\splitmode$ component refines the abstract automaton by splitting a given abstract mode. It implements Algorithm~\ref{code:splitmode} that is presented in Appendix~\ref{sec:mode_split_algorithm}. The ${\sf Cache}$ component stores the per-mode initial sets from which reachsets have been computed to avoid repeating computations and for fixed point checking. The ${\sf Reachability}$ ${\sf subroutine}$ component is the reachability tool being used for computing reachsets called by $\computereachset$ in the ${\sf Dynamics}$ ${\sf file}$.\label{fig:summary}} 
  \label{fig:arch}
\end{figure}

\section{Hybrid Automata definition}
\label{sec:hybridautomata}
In this section, we present a definition of hybrid automata~\cite{ACHH93,TIOAmon,Mitra07PhD} that $\ourtool$ constructs for the given scenario.
\begin{definition}
\label{def:hybrid-automaton}
A {\em hybrid automaton} is a tuple 
\begin{align}
\ha := \langle \stateset, \fparset, \initset, \finitpar, \edgeset, \guard, \reset, f\rangle, \text{  where}
\end{align}
\vspace{-0.15in}
\begin{enumerate}[label=(\alph*)]
\item $\stateset \subseteq \reals^n$ is the continuous state space, or simply the {\em state} space, and  $\fparset \subseteq \reals^d$ is the discrete state space, which we call the {\em parameter} or {\em mode} space (it is equal to $\parset$ in the scenario verification problem),
\item $\langle\initset, \finitpar\rangle \subseteq \stateset \times \fparset$ is a pair of a compact set of possible initial states and an initial mode ($\finitpar = \initmode$ in our setting),
\item $\edgeset \subseteq \fparset \times \fparset$ is a set of edges that specify possible transitions between modes, 
\item $\guard: \edgeset \rightarrow 2^\stateset$ defines the set of states at which a mode transition over an edge is possible,
\item $\reset: \stateset \times \edgeset \rightarrow 2^\stateset$ defines the possible updates of the state after a mode transition over an edge, and
\item $f: \stateset \times \fparset \rightarrow \stateset$ is the {\em dynamic} function that define the continuous evolution of the state in each mode. It is Lipschitz continuous in the first argument.
\end{enumerate}
\end{definition}

\section{Symmetries of dynamical systems}
\label{sec:parametrized_sys_symmetries}
In this section, we present formal definitions and sufficient conditions for symmetries of dynamical systems from the literature which are used in Section~\ref{sec:symdef}.

Let $\Gamma$ be a group of smooth maps acting on $\stateset$. 

\begin{definition}[Definition 2 in \cite{russo2011symmetries}]
	\label{def:symmetry}
 We say that $\gamma \in \Gamma$ is a symmetry of the ODE in Section~\ref{sec:problem} if its differentiable, invertible, and for any solution $\xi(\stateinst_0,\fparinst, \cdot)$, $\gamma (\xi(\stateinst_0,\fparinst, \cdot))$ is also a solution.
\end{definition} 




Coupled with the notion of symmetries, is the notion of equivarient dynamical systems. 
\vspace{-0.05in}
\begin{definition}[\cite{russo2011symmetries}]
	\label{def:equivariance_input}
	The dynamic function $f: \stateset \times \fparset \rightarrow \stateset$ is said to be $\Gamma$-equivariant if for any $\gamma \in \Gamma$, there exists $\rho: \fparset \rightarrow \fparset$ such that, 
	\begin{align}
	\label{eq:equivariance_condition}
	\forall\ \stateinst \in \stateset,\forall\ \fparinst \in \fparset,\ \frac{\partial \gamma}{\partial \stateinst} f(\stateinst, \fparinst) = f(\gamma(\stateinst), \rho(\fparinst)).
	\end{align}
\end{definition}

The following theorem draws the relation between symmetries and equivariance definitions. 
\begin{theorem}[Theorem 10 in \cite{russo2011symmetries}]
	\label{thm:sol_transform_parametrized_nonlinear}
	If $f$ is $\Gamma$-equivariant, then all maps in $\Gamma$ are symmetries of the ODE in Section~\ref{sec:problem}. Moreover, for any $\gamma \in \Gamma$, map $\rho: \fparset \rightarrow \fparset$ that satisfies equation (\ref{eq:equivariance_condition}), $x_0 \in \stateset$, and $\fparinst \in \fparset$, $\gamma(\xi(\stateinst_0, \fparinst,\cdot)) = \xi(\gamma(\stateinst_0), \rho(\fparinst), \cdot)$.  
\end{theorem} 

That means that one can get the trajectory of the system starting from an initial state $\gamma(\stateinst_0)$ in mode $\rho(\fparinst)$ by transforming its trajectory starting from $\stateinst_0$ in mode $\fparinst$ using $\gamma$.

\section{Abstraction definition}
\label{sec:abstraction_definition}
In this section, we define formally symmetry abstractions \cite{sibai-tac-2020}. The abstraction requires a set of symmetry maps which they call the {\em virtual map}, defined as follows:
\begin{definition}[Definition 4 in \cite{sibai-tac-2020}]
	\label{def:virtual_map}
	Given a hybrid automaton $\ha$, a virtual map is a set  
	\begin{align}
	\Phi = \{(\gamma_p, \rho_p) \}_{\fparinst \in \fparset},
	\end{align}
	where for every $\fparinst \in \fparset$,
	$\gamma_{\fparinst}: \stateset \rightarrow \stateset$, 
	$\rho_\fparinst: \fparset \rightarrow \fparset$, and they satisfy equation~(\ref{eq:equivariance_condition}).
\end{definition}

Given a virtual map $\Phi$, they define a new map $\rv: \fparset \rightarrow \fparset$ as follows: $\rv(\fparinst) = \rho_\fparinst(\fparinst)$. This is equal to our initialization of the map $\rv$ in line~\ref{ln:rv_definition} of Algorithm~\ref{code:bettersafetyVerifAlgo}.
It maps modes of the original automaton $\ha$ to their corresponding ones in the abstract one $\ha_v$, defined next. 

The idea of the abstraction is to group modes of $\ha$ that share similar behavior in symmetry terms together in the same mode in the abstract automaton. The trajectories of such modes can be obtained by transforming the trajectories of the corresponding abstract mode using symmetry maps. 

\begin{definition}[Definition 5 in \cite{sibai-tac-2020}]
	\label{def:hybridautomata_virtual}
	Given a hybrid automaton $\ha$,   
	and a virtual map $\Phi$,
	the resulting {\em abstract (virtual) hybrid automaton} is: 
	$$\ha_v = \langle \stateset_v, \fparset_v, \initset_v, \finitparv,  
	\edgeset_v, \guard_v, \reset_v, f_v \rangle \text{, where } $$
	\begin{enumerate}[label=(\alph*)]
		\item $\stateset_v = \stateset$ and 
		$\fparset_v = \rv(\fparset)$
		\label{item:def_virtual_mode_set}
		\item $\initsetv = \gamma_{\finitmode}(\initset)$ and $\initmodev = \rv(\finitmode)$, \label{item:def_virtual_initial}
		\item $\edgeset_v = \rv(\edgeset) = \{ (\rv(\fparinstvf),\rv(\fparinstvs))\ |\ \edgeinst = (\fparinstvf,\fparinstvs) \in \edgeset \}$
		\label{item:def_virtual_edgeset}
		\item $\forall \edgeinst_v \in \edgeset_v, $
		\begin{align*}
		\guard_v(e_v) = \bigcup\limits_{\edgeinst \in \rv^{-1}(\edgeinst_v)}  \gamma_{\edgeinst.\src}\big(\guard(\edgeinst)\big),
		\end{align*}
		\label{item:def_virtual_guard}
		\item $\forall \stateinst_v \in \stateset_v, \edgeinst_v \in \edgeset_v, $ 
		\begin{align*}
		\reset_v(\stateinst_v, e_v) = \bigcup\limits_{\edgeinst \in \rv^{-1}(\edgeinst_v)}\ \gamma_{\edgeinst.\dest}\big(\reset\big(\gamma_{\edgeinst.\src}^{-1}(x_v), \edgeinst \big)\big), \text{ and}
		\end{align*}
		\label{item:def_virtual_reset}
		\item 
		$\forall \fparinst_v \in \fparset_v$, $\forall \stateinst \in \stateset, f_v(\stateinst, \fparinst_v) = f(\stateinst, \fparinst_v)$. \label{item:def_virtual_dynamics}
	\end{enumerate}
\end{definition} 

\section{Split modes algorithm}
\label{sec:mode_split_algorithm}

In this section, we describe $\splitmode$, the algorithm used to split an abstract mode $\parinst_v^*$ and update the abstract automaton $\ha_v$. 
Then, we prove that its result $\ha_v'$ is an abstraction of $\ha$ and that $\ha_v$ is an abstraction of $\ha_v'$, using two forward simulation relations.

\subsubsection{$\splitmode$ description}

The procedure $\splitmode$ takes as input the mode to be split $\parinst_v^*$, the map $\rv$ that maps original modes to abstract ones, the original and abstract automata $\ha$ and $\ha_v$, and the original and abstract unsafe maps $\unsafeset$ and $\unsafeset_v$. It outputs the 
the updated map $\rv$, updated abstract automaton $\ha_v$, and the updated abstract unsafe map $\unsafeset_v'$.

It starts by obtaining the set of original modes $\parset^*$ of $\ha$ represented by $\parinst_v^*$ in line~\ref{split_ln:orig_modes}. If the number of these modes is less than two, then the algorithm returns the given input 
to indicate failure to split in line~\ref{split_ln:cantrefine}. Otherwise, it initializes the output variables $\rv'$, $\ha_v'$, and $\unsafeset_v'$ by copying $\rv$, $\ha_v$,and $\unsafeset_v$. It then creates two new abstract modes $\parinstvf$ and $\parinstvs$ and adds them to $\parset_v'$ in line~\ref{split_ln:createtwomodes}. Then, it decomposes $\parset$ into two disjoint sets $\parset_1^*$ and $\parset_2^*$, in line~\ref{split_ln:splitorigmodes}. After that, it updates the map $\rv$ to map modes in $\parset_1^*$ to $\parinstvf$ and those of $\parset_2^*$ to $\parinstvs$ in line~\ref{split_ln:updaterv}. It then updates the initial mode of $\ha_v'$ in case the split mode was the root.

Now that it created the new modes, it proceeds into updating the edges of $\ha_v'$ and their $\guard_v$ and $\reset_v$ annotations. It iterates over the edges that connect $\parinst_v^*$ with its {\em parents}, which may include $\parinst_v^*$ itself, and create for each such edge, two edges connecting that parent with both $\parinstvf$ and $\parinstvs$ in line~\ref{split_ln:foraddingparentsedges}. It repeats the same process but for the edges that connect $\parinst_v^*$ with its {\em children} in line~\ref{split_ln:foraddingchildrenedges}. Finally, it checks if $\parinst_v^*$ had an edge connecting it to itself, and if that is the case, creates two edges connecting each of $\parinstvf$ and $\parinstvs$ to themselves. 

In line~\ref{split_ln:updatingguardsandresets}, it annotates the created edges with their guards and resets in the same way Definition~\ref{def:hybridautomata_virtual} defined them, but using the update map $\rv'$. In lines~\ref{split_ln:settingdynamics} and \ref{split_ln:settingdynamics2}, it sets the dynamic function of both modes to be the same as that of $\parinst_v^*$. Finally, it deletes $\parinst_v^*$ with all edges connected to it in line~\ref{split_ln:removingpv}. It then initializes the unsafe maps for the newly created modes by decomposing the unsafe set of $\parinst_v^*$ into those of the two modes. It returns in line~\ref{split_ln:returning} the new modes, new $\newrv$ and $\haref$, and the unsafe map $\unsafeset_v'$.

 \begin{algorithm}
	\small
	\caption{$\splitmode$($ \parinst_v^*, \rv, \ha_v, \unsafeset_v, \ha, \unsafeset$)}
	\label{code:splitmode}
	\begin{algorithmic}[1]
		\label{split_ln:input}
		\State $\parset^* \gets \rv^{-1}(\parinst_v^*)$.  \label{split_ln:orig_modes}
		\If{$|P| < 2$} 
		\textbf{ return: } $\rv, \ha_v, \unsafeset_v$.  \label{split_ln:cantrefine} 
		\EndIf
		\State Create copies of $\rv$, $\hav$, and $\unsafeset_v$, and name them  $\newrv$, $\haref$, and $\unsafeset_v'$. \label{split_ln:createcopies}
		\State Create two new virtual modes $\parinstvf$ and  $\parinstvs$ and add them to $\parset_v'$. \label{split_ln:createtwomodes}
		\State Split $\parset^*$ in half to two sets $\parset_1^*$ and $\parset_2^*$. \label{split_ln:splitorigmodes}
		\State $\newrv[\parset_1^*] \gets \parinstvf,\ \newrv[\parset_2^*] \gets\parinstvs$. \label{split_ln:updaterv}
		\State $\initmodev' \gets \rv'[\initmode]$
		\For{$\edgeinst_v \in \edgeset_v$ such that $\edgeinst_v.\dest = \parinst_v^*$}  \label{split_ln:foraddingparentsedges}
		\State Create two new edges $ (\edgeinst_v.\src,\parinstvf)$
		and $(\edgeinst_v.\src,\parinstvs)$.
		\EndFor
		\For{$\edgeinst_v \in \edgeset_v$ such that $\edgeinst_v.\src = \parinst_v^*$} \label{split_ln:foraddingchildrenedges}
		\State Create two new edges $ (\parinstvf, \edgeinst_v.\dest)$
		and $(\parinstvs, \edgeinst_v.\dest)$.
		\EndFor
		\If{$\exists\ \edgeinst_v \in \edgeset_v$ such that $\edgeinst_v.\src = \edgeinst_v.\dest = \parinst_v^*$} \label{split_ln:ifaddingselfedges}
		\State Create two new edges  $ (\parinstvf, \parinstvf)$ and $ (\parinstvs, \parinstvs)$.
		\EndIf
		\State Define the guards and resets of new edges using the virtual map $\newrv$. \label{split_ln:updatingguardsandresets}
		\State Remove added edges that have empty guards.
		\State Set $f_v'(\cdot, \parinstvf) = f_v(\cdot, \parinst_v^*)$. \label{split_ln:settingdynamics}
		\State Set $f_v'(\cdot, \parinstvs) = f_v(\cdot, \parinst_v^*)$. \label{split_ln:settingdynamics2}
		\State Remove $\parinst_v^*$ from  $\parset_v'$ and $\unsafeset_v'$, and remove all attached edges from $\edgeset_v'$. \label{split_ln:removingpv}
		\State $\unsafeset_v'(\parinstvf) \gets \cup_{\parinst \in \parset_1} \gamma_{\parinst}(\unsafeset_v(\parinst))$
		\State $\unsafeset_v'(\parinstvs)\gets \cup_{\parinst \in \parset_2} \gamma_{\parinst}(\unsafeset_v(\parinst))$
		\State \textbf{return: } $ \newrv, \haref, \unsafeset_v'$ \label{split_ln:returning}
	\end{algorithmic}
\end{algorithm}

\subsubsection{Correctness guarantees of $\splitmode$}
\label{sec:fsr_refined}
In this section, we show that the resulting automaton $\haref$ from $\splitmode$ is still a valid abstraction of $\ha$, but it is a tighter one than $\ha_v$ by showing that $\ha_v$ is an abstraction of $\haref$.

Consider $\fsrref$, the same relation as $\fsr$ defined in Theorem 3 in \cite{sibai-tac-2020}, but using $\newrv$ instead of $\rv$. Formally, 
 $\fsrref \subseteq (\stateset \times \parset)\times (\stateset_v' \times \parset_v')$ defined as $(\stateinst, \parinst) \fsrref (\stateinst_v',\parinst_v')$ if and only if: 
\begin{enumerate}[label=(\alph*)]
	\item $\stateinst_v' = \gamma_p(\stateinst)$, and
	\item $\parinst_v' = \newrv(\parinst)$.
\end{enumerate}
Let us refer to $\fsrref$ by $\fsrf$ and let $\fsrs \subseteq (\stateset_v' \times \parset_v') \times (\stateset_v \times \parset_v)$ be defined as:  
$(\stateinst_v', \parinst_v') \fsrs (\stateinst_v,\parinst_v)$ if and only if: 
\begin{enumerate}[label=(\alph*)]
	\item $\stateinst_v = \stateinst_v'$, and
	\item \begin{align}
	\parinst_v =
	\begin{cases}
	\parinst_v', \mbox{ if }\parinst_v' \notin \{\parinstvf, \parinstvs\}, \\
	\parinst_v^*,\mbox{ otherwise.} 
	\end{cases}
	\end{align}
\end{enumerate}
The following theorem shows that these two relations are forward simulation relations between $\ha$ and $\haref$ and $\haref$ and $\ha_v$, respectively. 
\vspace{-0.05in}
\begin{theorem}
\label{thm:abstraction_correctness}
Fix any abstract mode $\parinst_v^* \in \parset_v$ of $\ha_v$, let $$\newrv, \haref, \unsafeset_v' = \splitmode(\parinst_v^*, \rv,\ha_v,\unsafeset_v,\ha,\unsafeset).$$ Then, the resulting relations $\fsrf$ and $\fsrs$ are FSRs from $\ha$ to $\haref$ and $\haref$ to $\ha_v$, respectively, and $\ha \preceq_{\fsrf} \haref \preceq_{\fsrs} \ha_v$.
\end{theorem}
\begin{proof}
Let us prove the first half first: that $\fsrf$ is a FSR from $\ha$ to $\haref$. We do that by showing that $\ha_v'$ is the result of following Definition~\ref{def:hybridautomata_virtual} to create an abstraction of $\ha$ using a slightly modified version $\Phi'$ of the virtual map $\Phi$, where  $\Phi'$ itself is another virtual map for $\ha$.

By definition, $\forall \parinst \in \parset, \rv(\parinst) = \rho_\parinst(\parinst)$, where $(\gamma_\parinst, \rho_\parinst) \in \Phi$. Let $\Phi'$ be equal to $\Phi$ for all $\parinst \notin \parset$, where $\parset$ is as in line~\ref{split_ln:orig_modes} of $\splitmode$. For any $\parinst \in \parset$, let $\rho'_\parinst(\parinst) = \parinstvf$, if $\parinst \in \parset_1$, and $\rho'_\parinst(\parinst) = \parinstvs$, otherwise. Moreover, as in lines~\ref{split_ln:settingdynamics} and \ref{split_ln:settingdynamics2}, define the continuous dynamics $f'_v$ to be equal to $f_v$, for all $\parinst_v \in \parset_v \textbackslash \{\parinstvf, \parinstvs\}$, and to be equal to $f_v(\cdot, \parinst_v^*)$, otherwise. Then, $\Phi'$ is a virtual map of $f'_v$ since any $(\gamma_\parinst', \rho_\parinst') \in \Phi'$ satisfies equation~(\ref{eq:equivariance_condition}) for $f'_v$, because the corresponding $(\gamma_\parinst, \rho_\parinst) \in \Phi$ satisfies it for $f_v$. The map $\newrv$ is just the result of $\Phi'$ as $\rv$ is the result of $\Phi$. 

The edges created in $\splitmode$ for $\parinstvf$ and $\parinstvs$ in $\ha_v'$ are a decomposition of the edges connected to $\parinst_v^*$ in $\ha_v$, including self edges. Hence, the output of $\splitmode$ $\ha_v'$ is indeed the result of following Definition~\ref{def:hybridautomata_virtual} to construct an abstraction of $\ha$ using $\newrv$. It follows from Theorem 3 in \cite{sibai-tac-2020}, that $\haref$ is an abstraction of $\ha$ and $\fsrf = \fsrref$ is a corresponding FSR.

Now we prove the second half of the theorem: that $\fsrs$ is a FSR from $\ha_v'$ to $\ha_v$. We follow similar steps of the proof of the first half in defining a new map, which we name $\Phi_2$, and prove that it is a virtual map of $\ha_v'$. Let $\Phi_2 = \{(\gamma_{\parinst_v'}, \rho_{\parinst_v'})\}_{\parinst_v' \in \parset_v'}$, where $\gamma_{\parinst_v'}(\stateinst_v') = \stateinst_v'$ is the identity map and $\rho_{\parinst_v'}(\parinst_v') = \parinst_v'$, if $\parinst_v' \notin \{\parinstvf, \parinstvs\}$, and $\rho_{\parinst_v'}(\parinst_v') = \parinst_v^*$, otherwise. Because of lines~\ref{split_ln:settingdynamics} and \ref{split_ln:settingdynamics2}, $(\gamma_{\parinst_v'}, \rho_{\parinst_v'})$ satisfy equation~(\ref{eq:equivariance_condition}) with the RHS dynamic function being $f_v'$ of $f$. Finally, notice that $\ha_v$ can be retrieved from $\haref$ by following Definition~\ref{def:hybridautomata_virtual} using $\Phi_2$. It follows from Theorem 3 in \cite{sibai-tac-2020}, that $\fsrs$ is a FSR from $\ha_v'$ to $\ha_v$. Thus,  $\ha \preceq_{\fsrf} \haref \preceq_{\fsrs} \ha_v$. 
\end{proof}

\subsubsection{Usefulness of $\splitmode$ in safety verification}
\label{sec:splittingusefulness}
We discuss now the benefits of splitting a mode $\parinst_v^*$, where $\Reach_{\ha_{v}}(\parinst_v^*) \cap \unsafeset_v(\parinst_v^*)\neq \emptyset$. The non-empty intersection with the unsafe set can mean either that:
\begin{enumerate}
\item {\em genuine counterexample}: $\exists \parinst \in \parset$ of $\ha$ where $\Reach_{\ha}(\parinst) \cap \unsafeset(\parinst) \neq \emptyset$, and thus $\ha$ is unsafe, or
\item {\em spurious counterexample}: there exists an execution of $\ha_v$ that does not correspond to a one of $\ha$ that is intersecting $\unsafeset_v(\parinst_v^*)$, and thus the intersection is a result of the abstraction, and not a correct counterexample.
\end{enumerate}  
Spurious counter examples could happen because of the guards and resets of the edges incoming to $\parinst_v^*$ being too large that the initial set of states for that mode is being larger than it should. Remember from Definition~\ref{def:hybridautomata_virtual}, that the guard and reset of any of these edges $\edgeinst_v$ is the union of all the transformed guards and resets of the edges in $\edgeset$ that get mapped to $\edgeinst_v$. If too many of the original edges are mapped to $\edgeinst_v$, its guard and reset will get larger, causing more transitions and larger initial set of $\parinst_v^*$ in $\ha_v$. This might increase the possibility of spurious counter example. Moreover, by definition, $\unsafeset_v(\parinst_v)$ is the union of the unsafe sets of all the modes that are mapped to $\parinst_v^*$ under $\rv$. The more the original modes that get mapped to $\parinst_v^*$, the larger is the unsafe set of $\parinst_v^*$ and the higher is the chance of a spurious counter example.  

Upon splitting $\parinst_v^*$ into two abstract modes, the guards of the edges incoming to $\parinst_v^*$, i.e. where $\parinst_v^*$ is a destination, will have their guards divided between the edges to $\parinstvf$ and $\parinstvs$. Additionally, the unsafe sets of $\parinst_v^*$ will be divided between $\parinstvf$ and $\parinstvs$. This will make the over-approximation of the behaviors of $\ha$ by $\ha_v$ get tighter and safety checking less conservative.

\section{NN-controlled quadrotor case study}
\label{sec:quadrotor-case-study}
In this section, we will describe a case study of a scenario having a planner, NN controller, and a quadrotor and model it as a hybrid automaton. We use the quadrotor model that was presented in \cite{ivanov2019verisig} along, its trained NN controller (see Appendix~\ref{sec:nonsymcontrol} on how we modify it to be rotation symmetric), and a RRT planner to construct its reference trajectories, independent of its dynamics.  

The dynamics of the quadrotor are as follows:
\begin{align}
\label{eq:quadrotor_dynamics}
\dot{q} := 
\begin{bmatrix}
\dot{p}_x^q \\
\dot{p}_y^q \\
\dot{p}_z^q \\
\dot{v}_x^q \\
\dot{v}_y^q \\
\dot{v}_z^q
\end{bmatrix}
=
\begin{bmatrix}
\dot{v}_x^q \\
\dot{v}_y^q \\
\dot{v}_z^q \\
g \tan \theta \\
- g \tan \phi \\
\tau - g
\end{bmatrix},
\dot{w} := 
\begin{bmatrix}
\dot{p}_x^w \\
\dot{p}_y^w \\
\dot{p}_z^w \\
\dot{v}_x^w \\
\dot{v}_y^w \\
\dot{v}_z^w
\end{bmatrix} = 
\begin{bmatrix}
b_x \\
b_y \\
b_z \\
0 \\
0 \\
0
\end{bmatrix},
\end{align}
where $q$ and $w$ are the states of the quadrotor and the planner reference trajectory representing their position and velocity vectors in the 3D physical space, respectively. The variables $\theta$, $\phi$, and $\tau$ represent the control inputs pitch, roll, and thrust, respectively, provided by the NN controller. The input to the NN controller is the difference between the quadrotor state and the reference trajectory: $q - w$.  The $g = 9.81 m/s^2$ is gravity and $b_x,b_y,$ and $b_z$ are piece-wise constant resulting in a piece-wise linear planner trajectory. In our case, these would be determined by the RRT planner as we will explain next. 

The NN controller has two hidden layers with 20 neurons each with $\mathit{tanh}$ activation units and a linear output layer. It acts as a classifier to choose from a set $\inputset \subset [-0.1, 0.1] \times [-0.1, 0.1] \times [7.81, 11.81]$ of eight possible control inputs. 
It was trained to mimic a model predictive control (MPC) controller to drive the quadrotor to follow the planner trajectory. A NN is used for its faster runtime computation and reachability analysis and smaller memory requirements than traditional MPC controllers.

Given an initial set of positions $K \subset \reals^3$, a goal set of positions $\goalset = \cup_i  \goalset_i \subset \reals^3$, and a set of 3D obstacles, the planner would generate a directed graph over $\reals^3$ that connects $K$ to every $G_i$ with piece-wise linear paths. We denote the set of linear segments in the graph by $\roads:= \{\roadinst_i\}_i$. The planner ensures that the waypoints and segments do not intersect obstacles, but without regard of the quadrotor dynamics. 

The ${\sf Hybrid}$ ${\sf constructor}$ in $\ourtool$ models such a scenario as  a hybrid automaton: 
\begin{enumerate}[label=(\alph*)]
\item $\stateset = \reals^6$, the space in which the state of the quadrotor $q$ lives, and $\parset = \roads$, the space in which the graph segments live, where the first three and last three coordinates  determine the start and end points $\parinst.\src$  and $\parinst.\dest$ of the segment, respectively,
\item $\langle \Theta, \initpar \rangle := \langle [K,[-0.5, 0.5]^3], \initpar \rangle$, where $[-0.5, 0.5]^3$ are the range of initial velocities of the quadrotor and $\initpar$ is the initial segment going out of $K$,
\item $\edgeset := \{(\roadinst_i, \roadinst_{i+1})\ |\ \roadinst_i, \roadinst_{i+1} \in \roads,\  \roadinst.\dest = \roadinst_{i+1}.\src\}$, 
\item $\guard((\roadinst_i, \roadinst_{i+1}))$ is the 6D ball centered at $[\roadinst_i.\dest, 0, 0, 0]$ with radius $[1,1,1,\infty,\infty,\infty]$, meaning that the quadrotor should arrive within distance 1 unit of the destination waypoint of the first segment, which is equivalent to the source waypoint of the second segment, at any velocity, to be able to transition to the second segment/mode,
\item $\reset(q, (\roadinst_i, \roadinst_{i+1})) = q$, meaning that there is no change in the quadrotor state after it starts following a new segment, and
\item $f(q, \roadinst) = g(q, h(q, \roadinst_i))$, where $g: \stateset \times \inputset \rightarrow \stateset$ is the right hand side of the differential equation of $q$ in equation~(\ref{eq:quadrotor_dynamics}) and $h: \stateset \times \parset \rightarrow \inputset$ is the NN controller. Without loss of generalization, we assume that $[b_x,b_y,b_z] \in \{-0.125, 0.125\}^3$. $b_x$ is equal to $-0.125$ if $\roadinst.\src[0] > \roadinst.\dest[0]$ and $0.125$ otherwise. The same applies for $b_y$ and $b_z$. \label{eq:quadrotor_dynamic_function}
\end{enumerate}

\section{Symmetry with non-symmetric controllers}
\label{sec:nonsymcontrol}

\subsection{Symmetries of closed loop control systems}
In this section, 
we discuss the property that the controller should satisfy for a closed-loop control system to be symmetric.

Fix an input space $\inputset\subseteq \reals^m$ and consider a right hand side of the ODE in Section~\ref{sec:problem} of the form:
\begin{align}
\label{eq:control-system}
f_c(\stateinst, \parinst) := g(\stateinst, h(\stateinst, \parinst)), 
\end{align}
where $g: \stateset \times \inputset \rightarrow \stateset$ and $h: \stateset \times \parset \rightarrow \inputset$ are Lipschitz continuous functions with respect to both of their arguments.

In order to retain symmetry for such systems, we update the notion of equivariance of dynamic functions. But first, let us define symmetric controllers.
\begin{definition}
	\label{def:symmetric_controller}
Given three maps  $\beta: \inputset \rightarrow \inputset$, $\gamma: \stateset \rightarrow \stateset$, and $\rho: \parset \rightarrow \parset$. We call the control function $h$,  $(\beta,\gamma,\rho)$-symmetric, if for all $\stateinst \in \stateset$ and $\parinst \in \parset$, $\beta (h(\stateinst, \parinst)) = h(\gamma(\stateinst), \rho(\parinst))$.
\end{definition}
Definition~\ref{def:symmetric_controller} means that if we transform the input of the controller, the state and the mode, using the maps $\gamma$ and $\rho$, respectively, then its output gets transformed with the map $\beta$. Such a property formalizes intuitive assumptions about controllers in general. For example, translating the position of the quadrotor and the planned trajectory by the same vector should not change the controller output. The NN controller discussed in Appendix~\ref{sec:quadrotor-case-study} indeed satisfies this property since its input is the relative state $q - w$.
We update the notion of equivariance for closed-loop control systems to account for the controller in the following definition.
\begin{definition}
\label{def:control-systems-equivarince}
We call the control system dynamic function $f_c$ of equation~(\ref{eq:control-system}) $\Gamma$-equivariant if for any $\gamma \in \Gamma$, there exist $\rho: \parset \rightarrow \parset$ and $\beta: \inputset \rightarrow \inputset$ such that $h$ is $(\beta,\gamma,\rho)$-symmetric and 
\begin{align}
\label{eq:control_equivariance_condition}
\forall\ \stateinst \in \stateset,\forall\ \inputinst \in \inputset,\ \frac{\partial \gamma}{\partial \stateinst} g(\stateinst, \inputinst) = g(\gamma(\stateinst), \beta(\inputinst)).
\end{align} 
\end{definition}

The following theorem repeats the results of Theorem~\ref{thm:sol_transform_parametrized_nonlinear} for the closed loop control system.

\begin{theorem}
	\label{thm:sol_transform_input_nonlinear}
	If $f_c$ of equation~(\ref{eq:control-system}) is $\Gamma$-equivariant, then all maps in $\Gamma$ are symmetries. Moreover, for any $\gamma \in \Gamma$, maps $\rho: \parset \rightarrow \parset$ and $\beta: \inputset \rightarrow \inputset$ that satisfy equation (\ref{eq:control_equivariance_condition}), $x_0 \in \stateset$, and $\parinst \in \parset$, $\gamma(\xi_c(\stateinst_0, \parinst,\cdot)) = \xi_c(\gamma(\stateinst_0), \rho(\parinst), \cdot)$, where $\xi_c$ is the trajectory of the dynamical system with RHS equation~(\ref{eq:control-system}). 
\end{theorem}
\begin{proof}
Fix an initial state $\stateinst_0 \in \stateset$, a mode $\parinst \in \parset$, and $\gamma \in \Gamma$ with its corresponding maps $\rho$ and $\beta$ that satisfy Definition~\ref{def:equivariance_input} per the assumption of the theorem. For any $t \geq 0$, let $\stateinst = \xi_c(\stateinst_0,\parinst, t)$ and $y = \gamma(\stateinst)$. Then, 
\begin{align}
\frac{d y}{dt} &= \frac{\partial \gamma}{\partial x} \frac{d\state}{dt}, \mbox{ using the chain rule,} \nonumber\\
&= \frac{\partial \gamma}{\partial x} g(\state,h(\state,\parinst)), \mbox{ using equation~(\ref{eq:control-system}),} \nonumber\\
&= g(\gamma(\state), \beta(h(\state,\parinst))),  \mbox{ using equation~(\ref{eq:control_equivariance_condition}),} \nonumber\\
&= g(\gamma(\state), h(\gamma(\state), \rho(\parinst))),  \mbox{using Definition~\ref{def:symmetric_controller},}\nonumber \\
&= g(y, h(y, \rho(\parinst))), \mbox{ by substituting $\gamma(x)$ with $y$,} \nonumber \\
&= f_c(y, \rho(\parinst)).
\end{align}
Hence, $\gamma( \xi_c(\stateinst_0,\parinst, t))$ also satisfies equation~(\ref{eq:control-system}) and thus a valid solution of the system. Therefore, $\gamma$ is a symmetry per Definition~\ref{def:symmetry}. Moreover, $y$ is a solution starting from $\gamma(\stateinst_0)$ in mode $\rho(\parinst)$. 
This proof is similar to that of Theorem 10 in \cite{russo2011symmetries} with is the difference of having a controller $h$, which requires the additional assumption that $h$ is symmetric.
\end{proof}

In Appendix~\ref{sec:fromnontosym}, we discuss how to make non-symmetric controllers symmetric, and apply that to the NN-controller of the quadrotor to make it rotation symmetric.

\subsection{From non-symmetric controllers to symmetric ones}
\label{sec:fromnontosym}
In some cases, the controller $h$ is not symmetric. For example, the NN controller of the quadrotor in Section~\ref{sec:ex} and Appendix~\ref{sec:quadrotor-case-study} is not symmetric with respect to rotations in the $xy$-plane. We show a counter example in Figure~\ref{fig:non-symmetric-ctrl}. 

\setlength{\textfloatsep}{0pt}
\begin{figure}[!htbp]
\nocaption
	\begin{subfigure}[t]{0.49\textwidth}
		\centering
		\includegraphics[width=\textwidth]{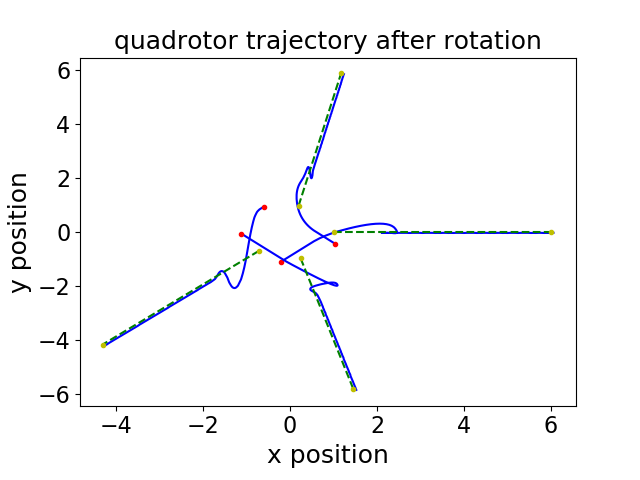}
		\vspace{-0.15in}
	    \caption{NN controller of the quadrotor is not rotation symmetric.\label{fig:non-symmetric-ctrl}
		}
		\vspace{\floatsep}
	\end{subfigure}
	\hspace{0.1in}
	\begin{subfigure}[t]{0.49\textwidth}
		\centering
		\includegraphics[width=\textwidth]{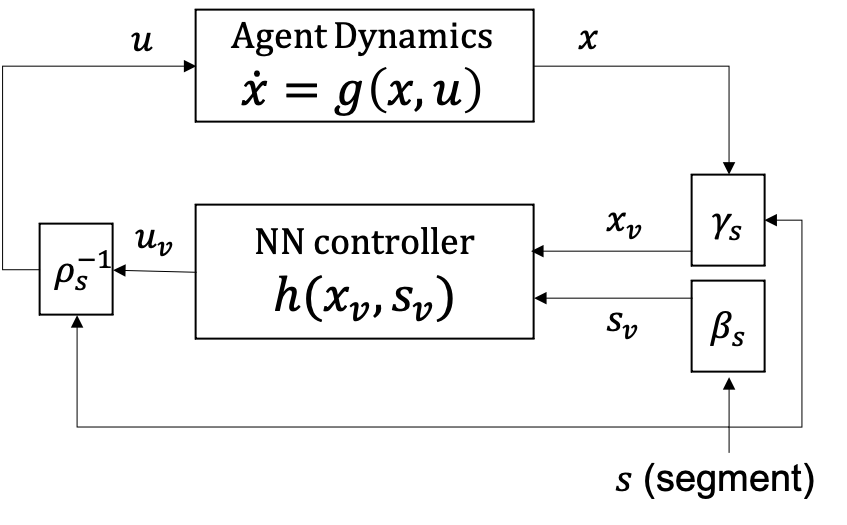}
		\vspace{-0.15in}
		\caption{Patching the controller to make it symmetric.\label{fig:patching_control}}
		\vspace{\floatsep}
	\end{subfigure}
\end{figure}

The NN input is the relative state $q - w$, which as we mentioned before, makes it symmetric to translations of the state and the reference trajectory. But, if we rotate the coordinate system of the physical $xy$-plane, 
there is no guarantee that the NN will change its outputs $\theta$ and $\phi$, such that the RHS of $\dot{v}_x^q$, and $\dot{v}_y^q$, $g \tan \theta$ and $-g \tan \phi$, are rotated accordingly.

Such non-symmetric controllers will prevent the dynamics from being equivariant. Equivariance is a desirable, and expected, property of certain dynamical systems. For example, vehicles dynamics are expected to be translation and rotation invariant in the $xy$-plane. Thus, non-symmetric controllers violate intuition about systems dynamics. Such controllers may not be feasible to abstract using Definition~\ref{def:hybridautomata_virtual}. 
Next, we will suggest a way to make any controller, including NN ones, such as that of the quadrotor, symmetric, leading to better controllers and retrieving the ability to construct abstractions. 

Consider again the closed loop control system of equation~(\ref{eq:control-system}). Let $\Phi = \{(\gamma_\parinst, \rho_\parinst)\}_{\parinst \in \parset}$ be a set of maps that we want it to be a virtual map of system~(\ref{eq:control-system}).  As before let us define $\rv: \parset \rightarrow \parset$ by $\rv(\parinst) = \rho_{\parinst}(\parinst)$, for all $\parinst \in \parset$. 
Assume that for every $\parinst \in \parset$, there exists $\beta_p: \inputset \rightarrow \inputset$, such that the open loop dynamic function $g$ is symmetric in the sense that it satisfies equation~(\ref{eq:control_equivariance_condition}). Moreover, assume that $\gamma_{\rv(\parinst)}, \beta_{\rv(\parinst)},$ and $\rho_{\rv(\parinst)}$ are identity maps for any $\parinst\in \parset$. This assumption means that applying the same symmetry transformation twice would not change the state nor the mode. 
Now, let us define a new controller $h': \stateset \times \parset \rightarrow \inputset$, that is shown in Figure~\ref{fig:patching_control}, as follows:
\begin{align}
\label{eq:symmetric_controller_definition}
	h'(\stateinst,\parinst) = 
	\beta_\parinst^{-1}(h(\gamma_\parinst(\stateinst),\rv(\parinst))).
\end{align}

\begin{theorem}
	\label{thm:making_controller_symmetric}
For any $\parinst \in \parset$, the controller $h'$ is $(\beta_\parinst, \gamma_\parinst, \rv)$-symmetric. 
\end{theorem}
\begin{proof}
Fix $\stateinst \in \stateset$ and $\parinst \in \parset$. Then, $\beta_\parinst (h'(\stateinst, \parinst)) $
\begin{align}
&= \beta_\parinst (\beta_\parinst^{-1} (h(\gamma_\parinst(\stateinst), \rv(\parinst))), \mbox{ using equation~(\ref{eq:symmetric_controller_definition}), } \nonumber \\
&= h(\gamma_\parinst(\stateinst), \rv(\parinst)), \mbox{ since $\beta_\parinst\beta_\parinst^{-1}$ is the identity map,} \nonumber \\
&= \beta_{\rv(\parinst)}^{-1}h(\gamma_{\rv(\parinst)}(\gamma_\parinst(\stateinst)), \rho_{\rv(\parinst)}(\rv(\parinst))), \nonumber \\
&\hspace{0.1in} \mbox{[since $\gamma_{\rv(\parinst)}$ and $\beta_{\rv(\parinst)}$ are identity maps, $\rho_{\rv(\parinst)}(\rv(\parinst)) = \rv(\parinst)$], } \nonumber \\
&= h'(\gamma_{\parinst}(\stateinst), \rv(\parinst)),
\end{align}
where the last equality follows from using equation~(\ref{eq:symmetric_controller_definition}) again.
\end{proof}
\vspace{-0.05in}
The controller $h'$ ensures that all modes $\parinst \in \parset$ that get mapped to the same mode $\parinst_v$ by $\rv$ have a transformed version, using $\beta_\parinst$, of a unique control for the same transformed state $\gamma_\parinst(\stateinst)$. That unique control is equal to $h$ with mode $\parinst_v$. This ensures that all modes that are equivalent under $\rv$ have symmetric behavior when the open loop dynamic function $g$ is symmetric as well.  

\section{Machine Specification}
\label{sec:machine_specification}
\begin{itemize}
    \item Processor:  AMD Ryzen 7 5800X CPU @ 3.8GHz x 8
    \item Memory: 32GB
\end{itemize}

\section{Trying other NN-controlled systems' verification tools as reachability subroutines}
\label{sec:other_tools}
\begin{enumerate}
\item The state-of-the-art verification tool for NN-controlled systems Verisig needs up to 30 minutes to compute the reachsets for 4 segments in a quadrotor scenario~\cite{ivanov2019verisig}. We tried Verisig and it took similar or longer amount of time for scenarios with less than five segments. We decided to use DryVR for faster evaluation of $\ourtool$ in this paper.
\item We tried NNV~\cite{NNV}, however the resulting reachsets had large over-approximation errors in our scenarios to the point of being not useful.  We contacted its developers and they are working on the conservativeness problem.
\item We considered using the tool of “Reachability analysis for neural feedback systems using regressive polynomial rule inference”, by Dutta et. al.~\cite{Sherlock-poly-2019}. We were unable to implement our scenarios in manageable time given that there is no manual to use the tool. 
\end{enumerate}

\newpage
\section{Reachset and Scenarios Figures}
\label{sec:reachset_figures}
\begin{figure}[H]
    \includegraphics[width=0.5\textwidth]{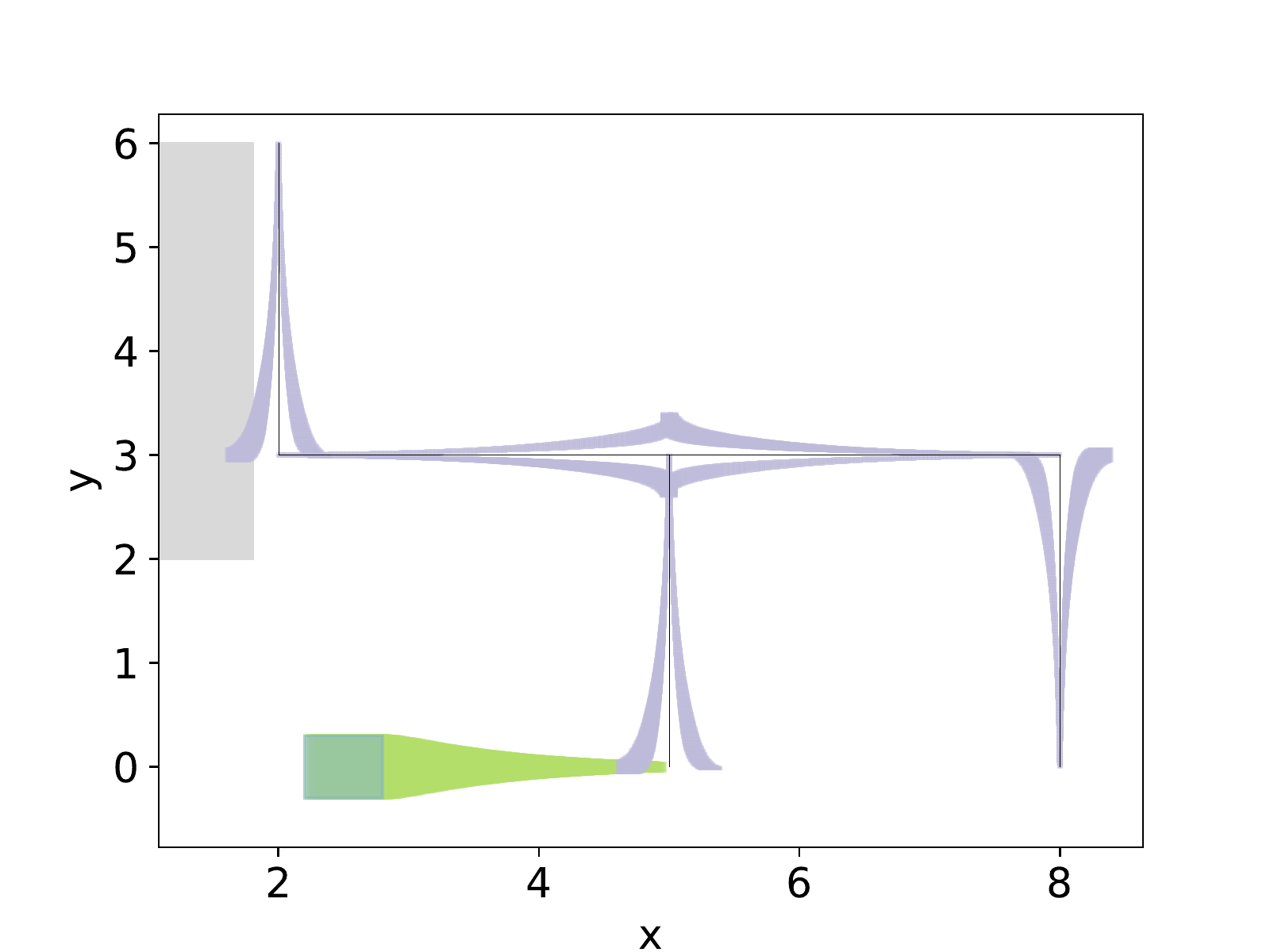}
    \includegraphics[width=0.5\textwidth]{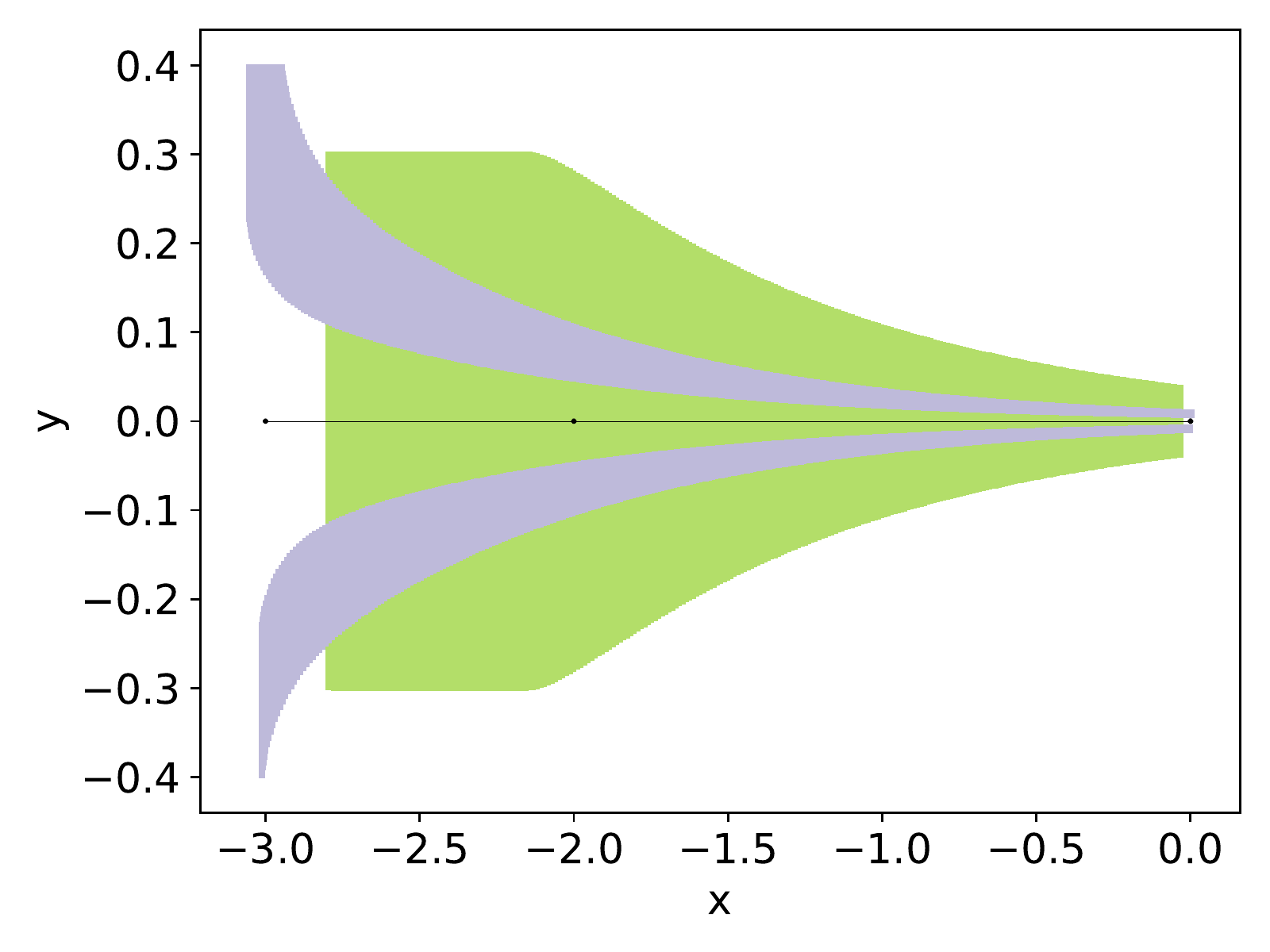}
    \includegraphics[width=0.5\textwidth]{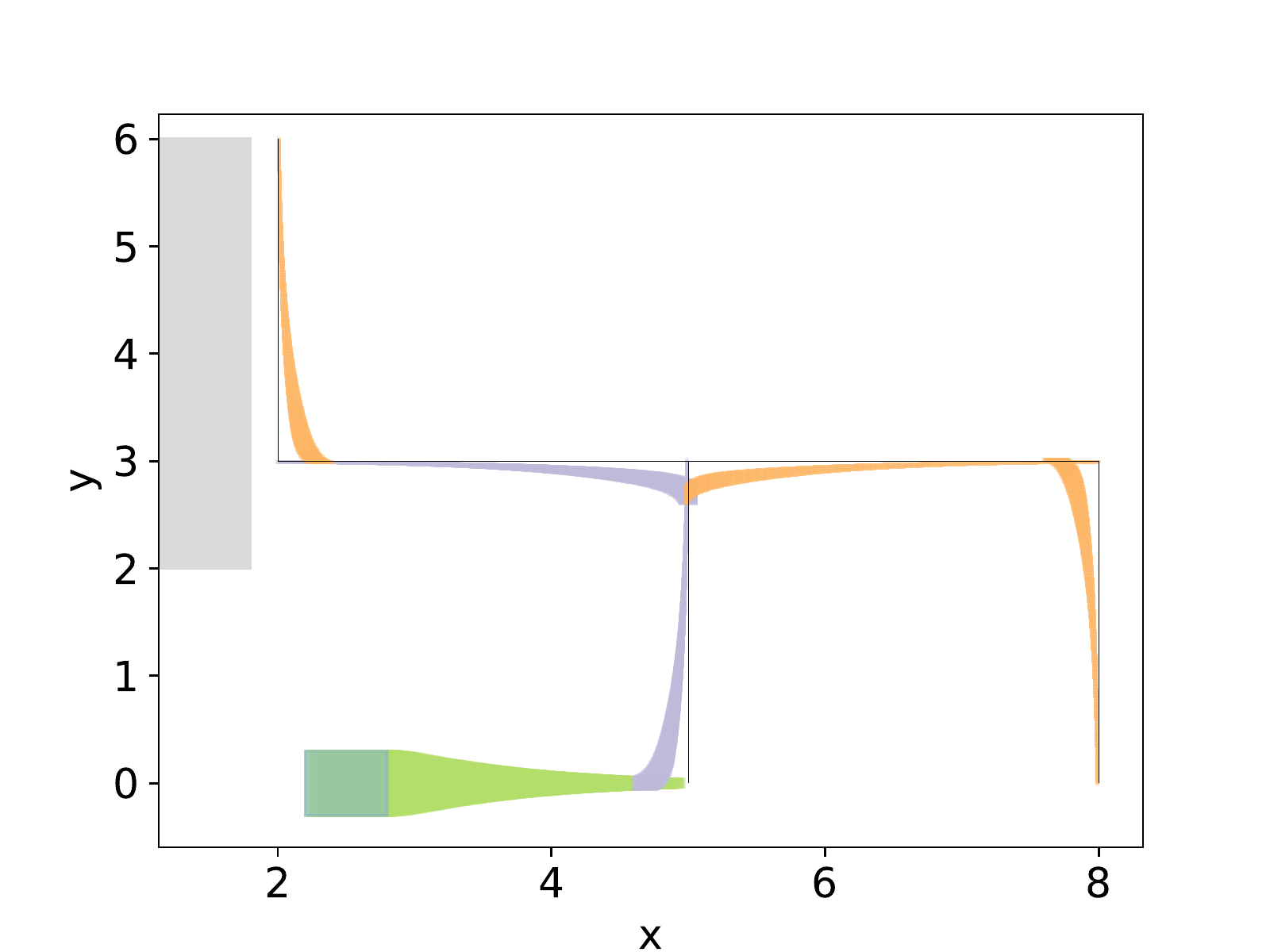}
    \includegraphics[width=0.5\textwidth]{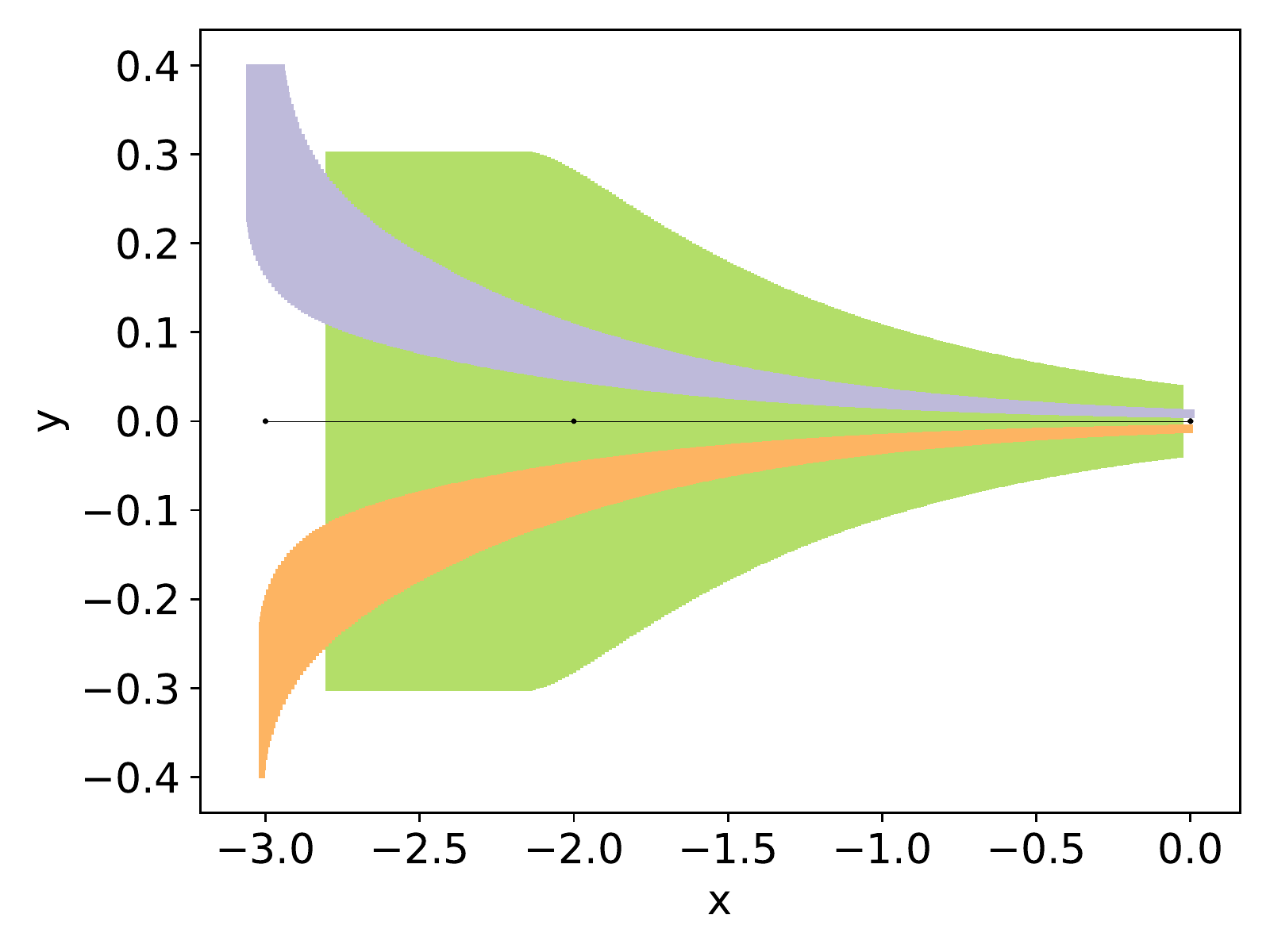}
    \includegraphics[width=0.5\textwidth]{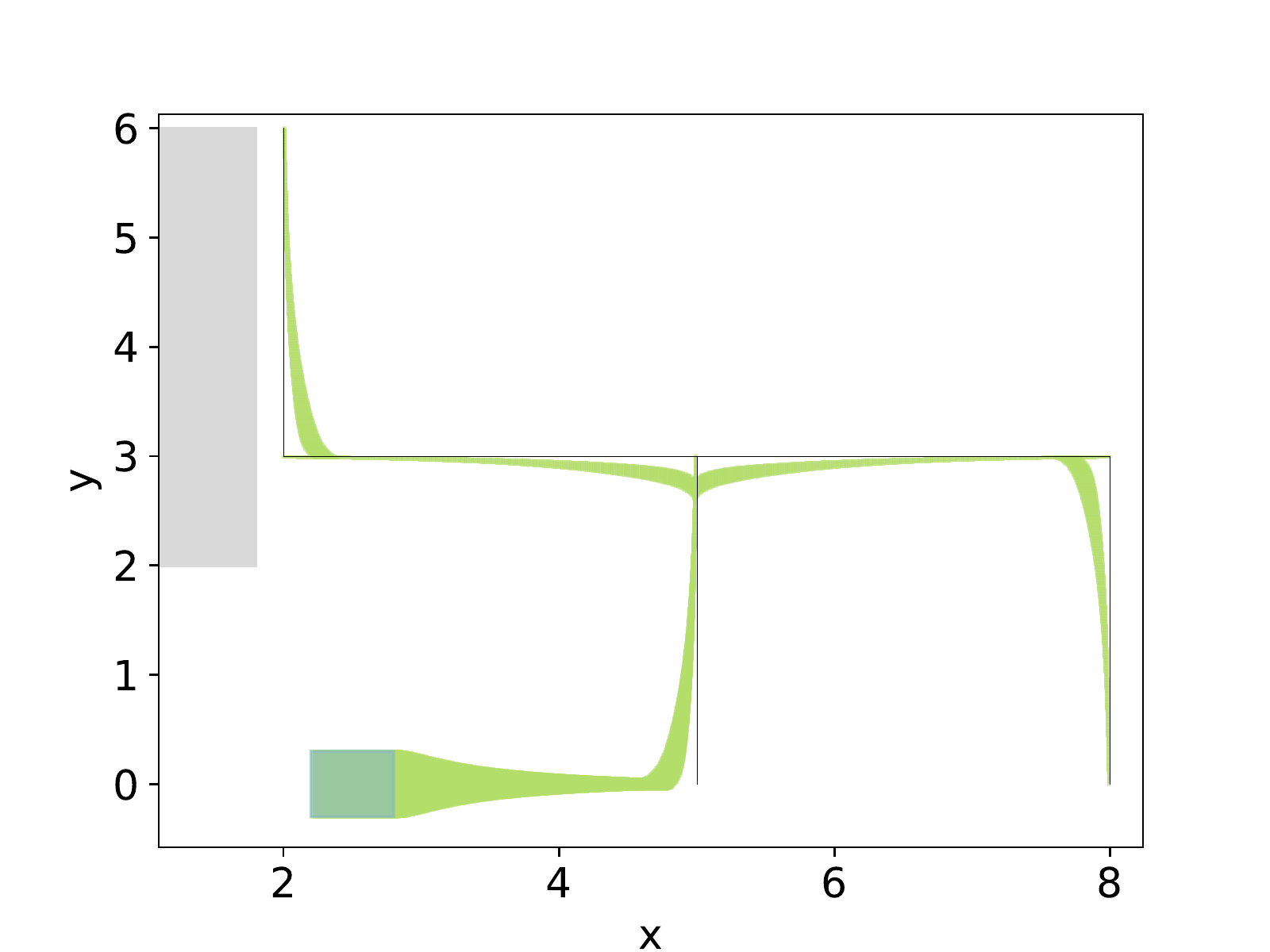}
    \caption{Scenario verification example C-S1. Uncertain initial positions (blue square). Plan defined by the black segments. Obstacle (grey rectangle). The possible positions the car might pass by (i.e. reachset) while following the plan (green, orange, and violet). The different colors correspond to different abstract segments (defined in Section~\ref{sec:symmetry_abstractions}). All figures are generated by $\ourtool$'s instance with DryVR subroutine using $\Phi =$ TR. The first row is when using abstraction but no refinements are allowed. The second row is where using the abstraction-refinement algorithm (only one refinement needed). The third row is generated without using the abstraction refinement algorithm. The left column represents the concrete scenario with the computed reachsets. This reachset is just used for visualization purposes but not used in the verification process. Only the abstract reachset is used for verification. The right column represents the reachsets of the abstract automaton.}
\label{fig:sceneverif}
\end{figure}

\begin{figure}[H]
    \includegraphics[width=0.49\textwidth]{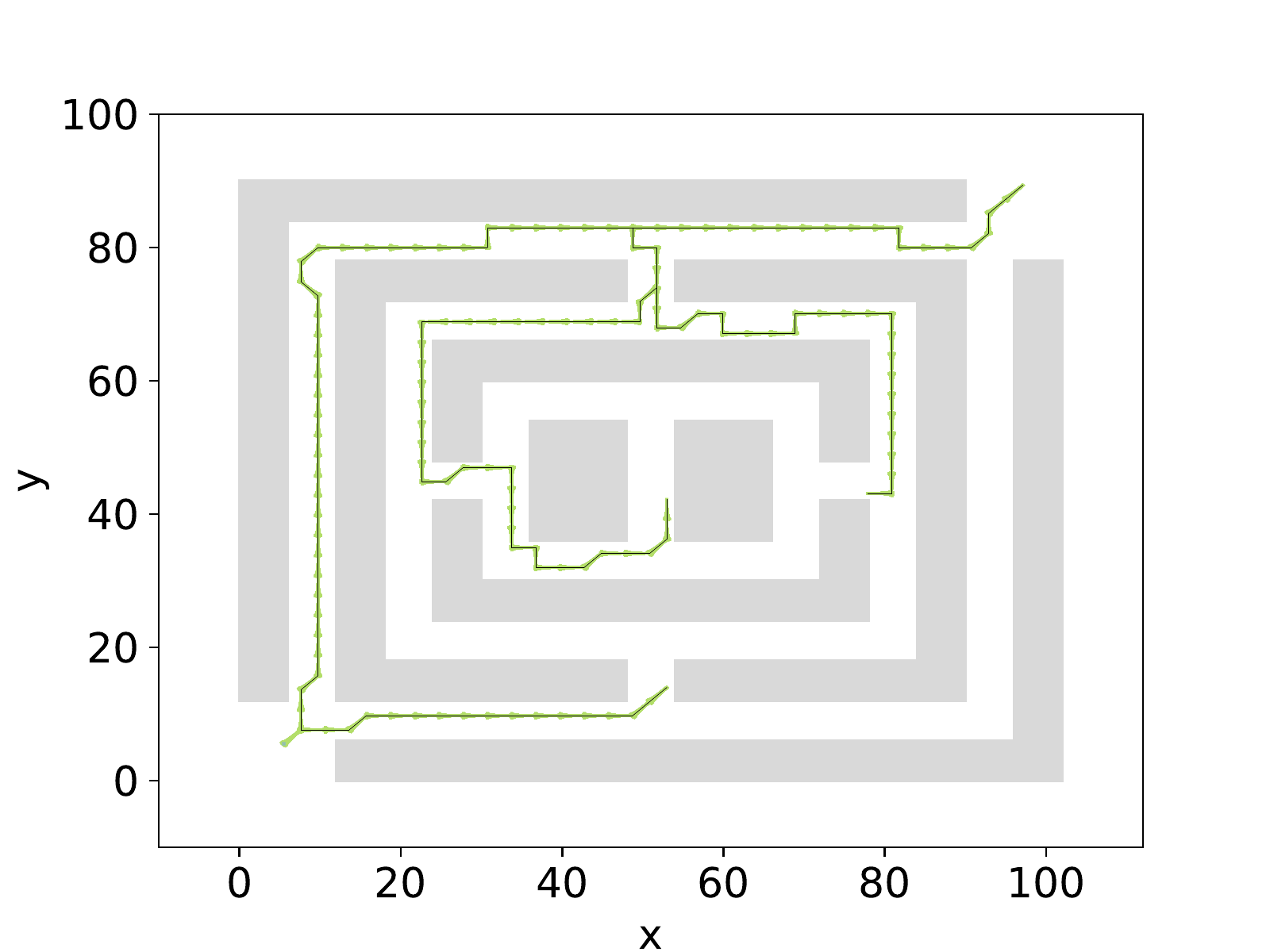}
    \includegraphics[width=0.49\textwidth]{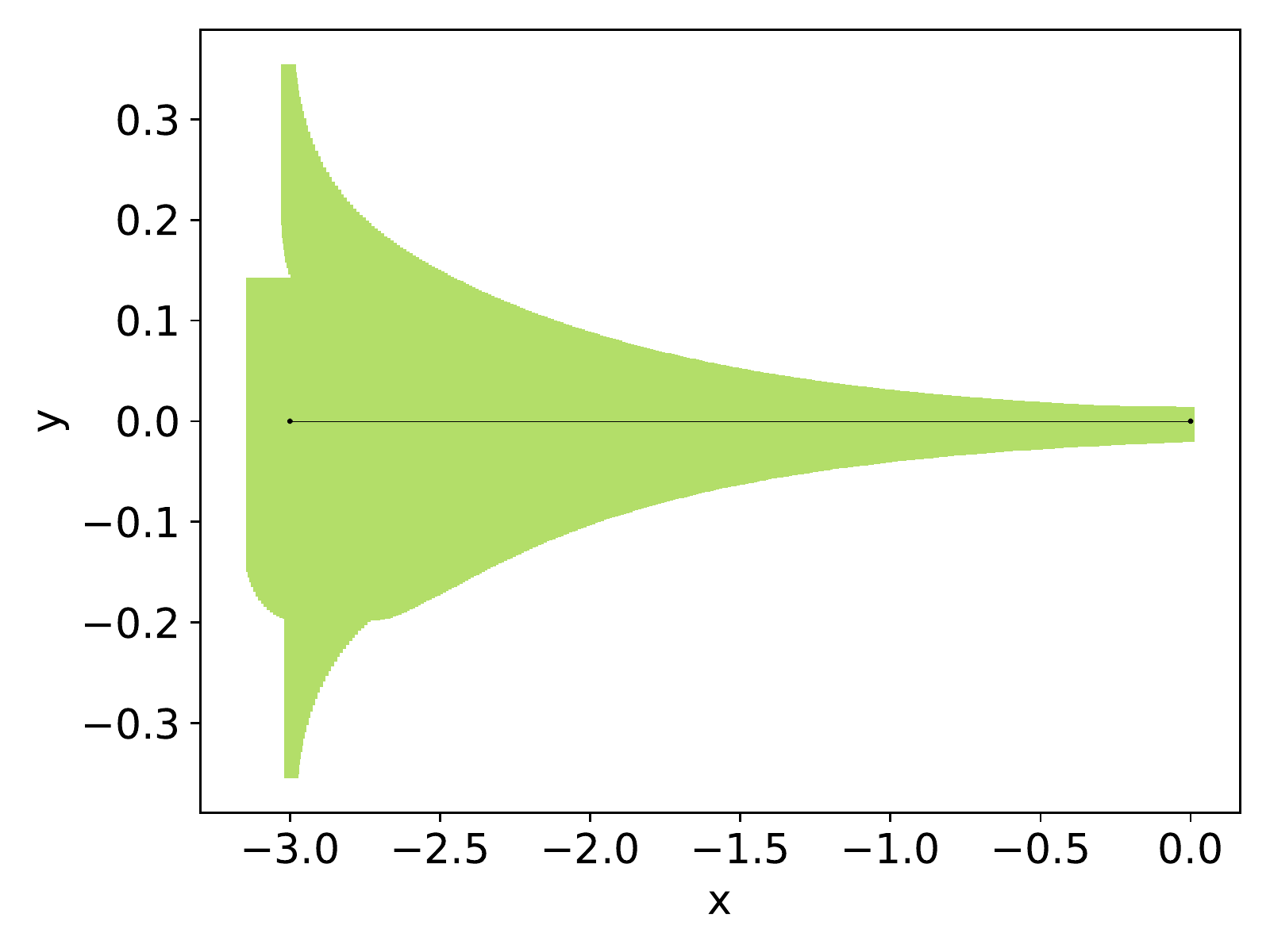}
    \includegraphics[width=0.49\textwidth]{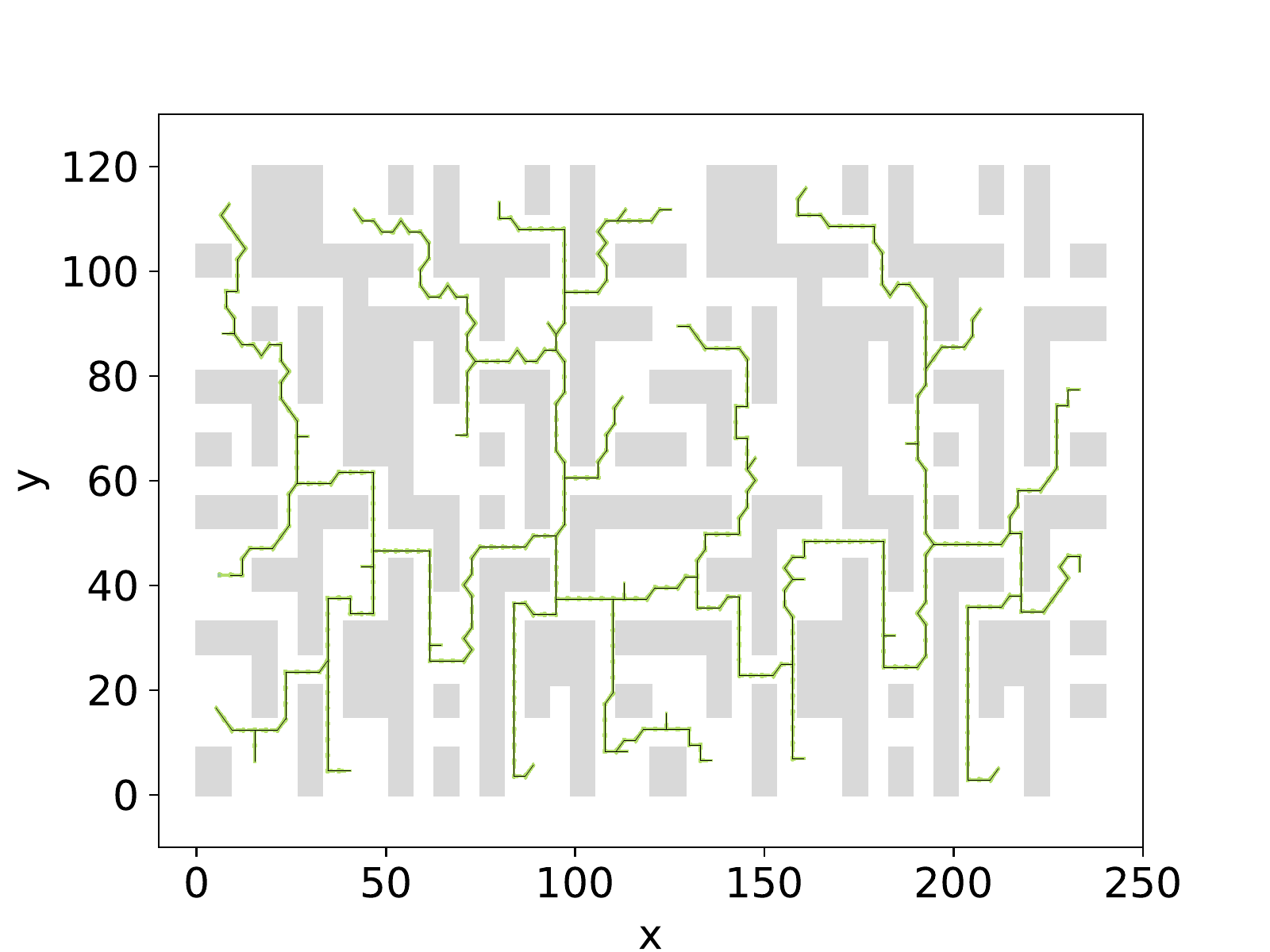}
    \includegraphics[width=0.49\textwidth]{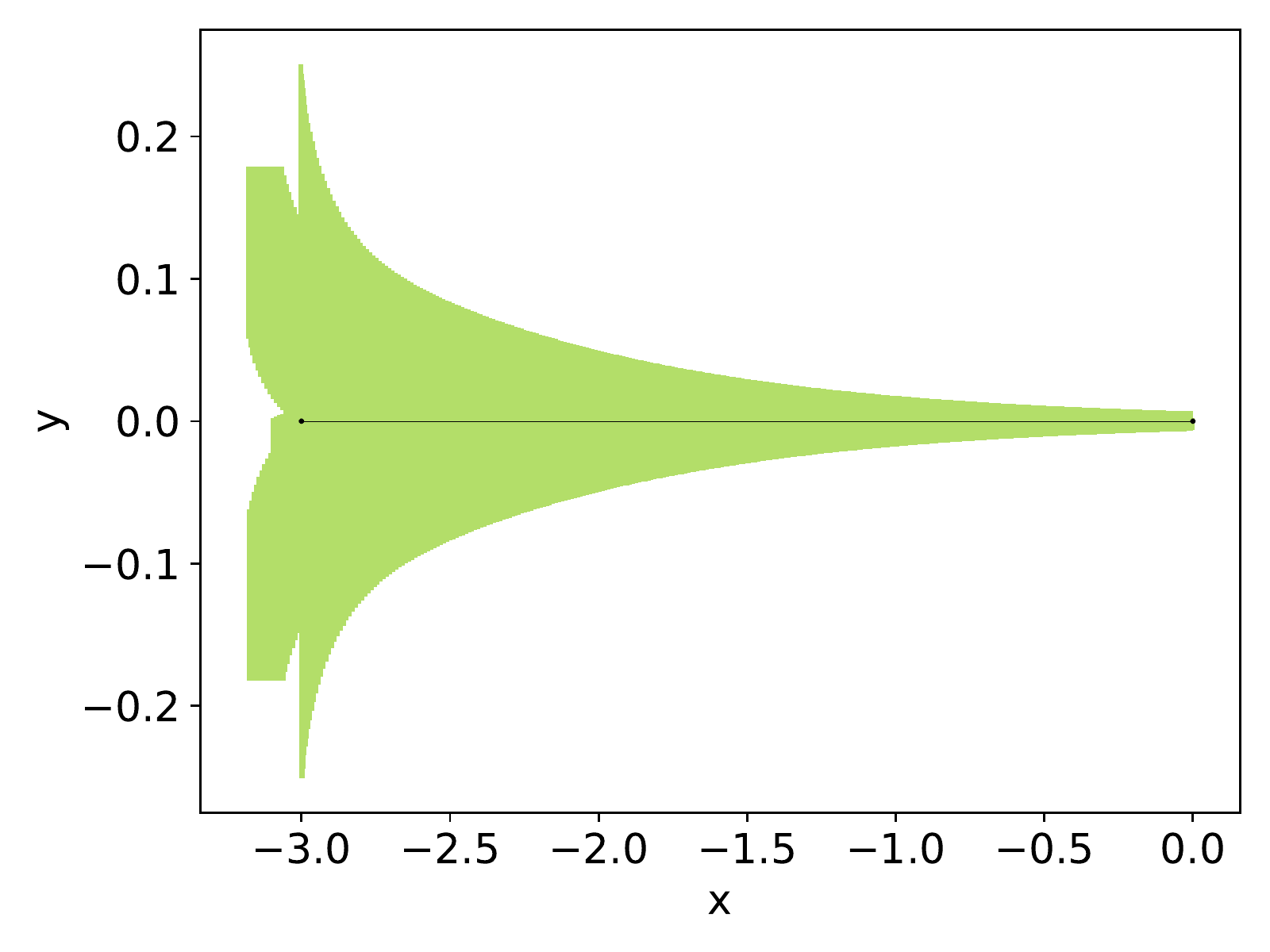}
    \includegraphics[width=0.49\textwidth]{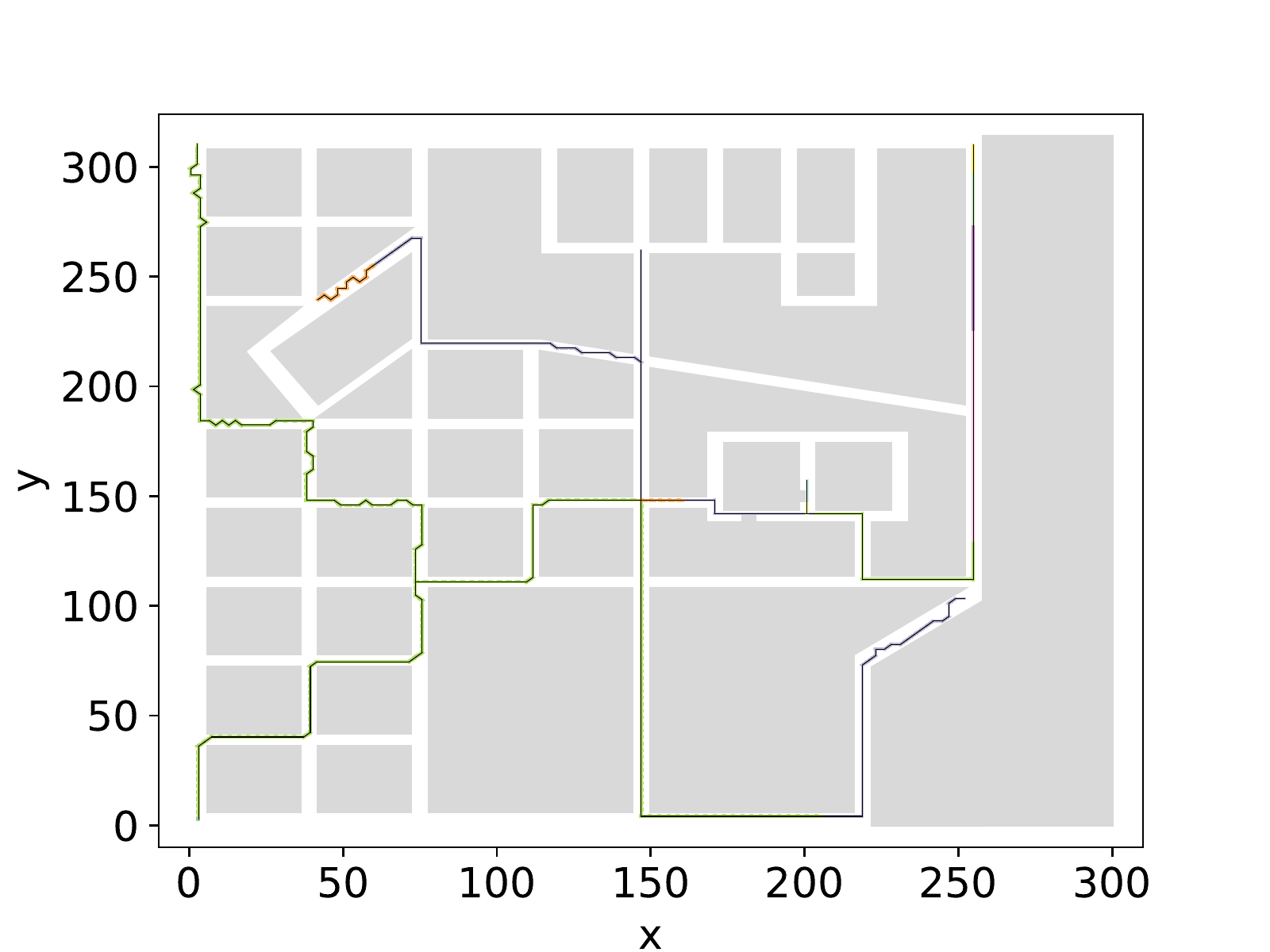}
    \includegraphics[width=0.49\textwidth]{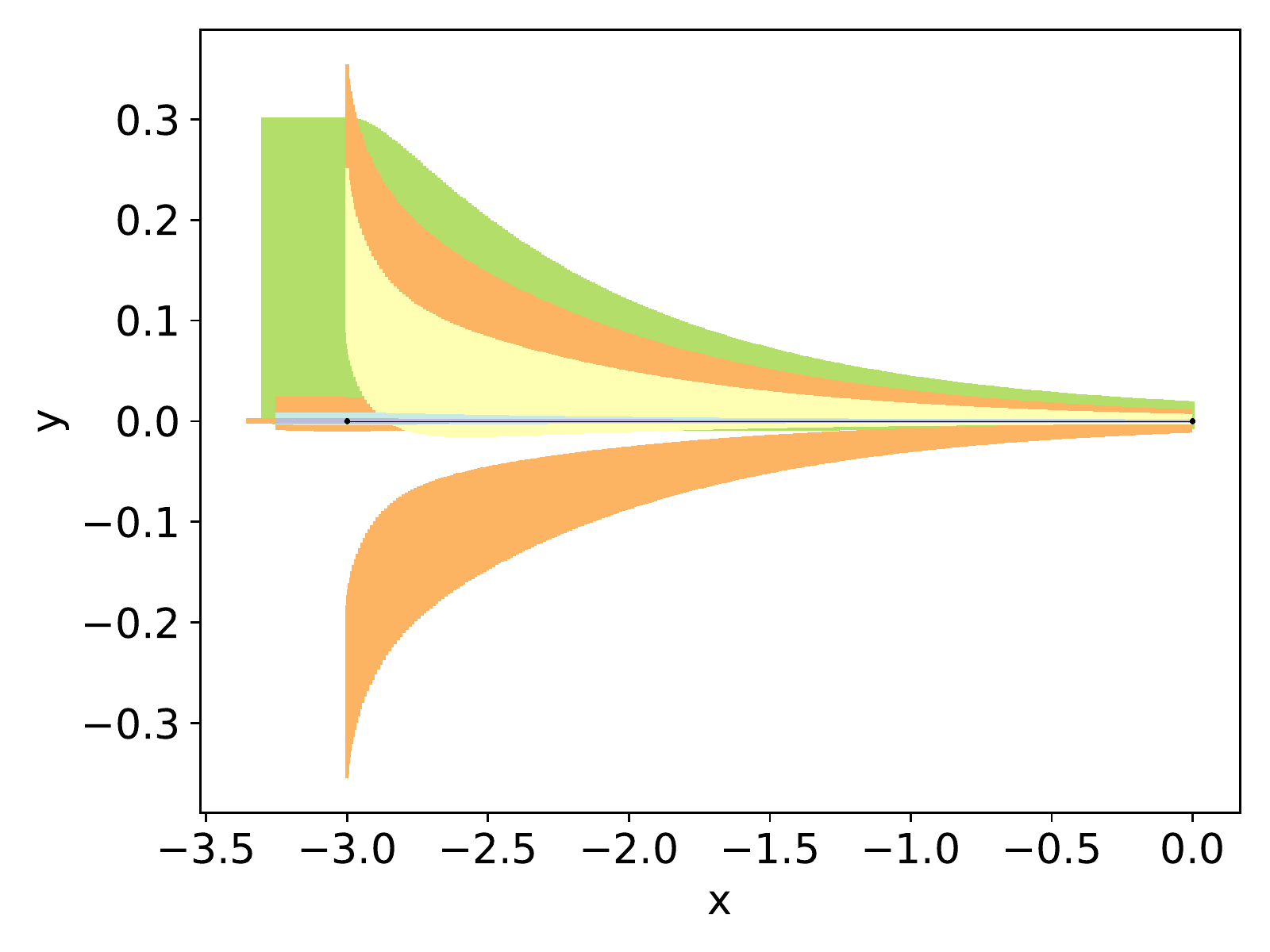}
    \caption{Reachsets of the car with the PD-controller in scenarios S2 (first row), S3 (second row), S4.b (third row). The left column represents the concrete scenario with the car reachsets. The right column represents the reachsets of the abstract automaton. The different colors correspond to different abstract segments (defined in Section~\ref{sec:symmetry_abstractions}). All abstract segments have the same waypoints but represent different segments of the concrete scenario. That means they lead to different reset of the agent's state after mode transitions.}
    \label{fig:car_senarios}
\end{figure}

\begin{figure}[H]
  \begin{center}
    \includegraphics[width=0.49\textwidth]{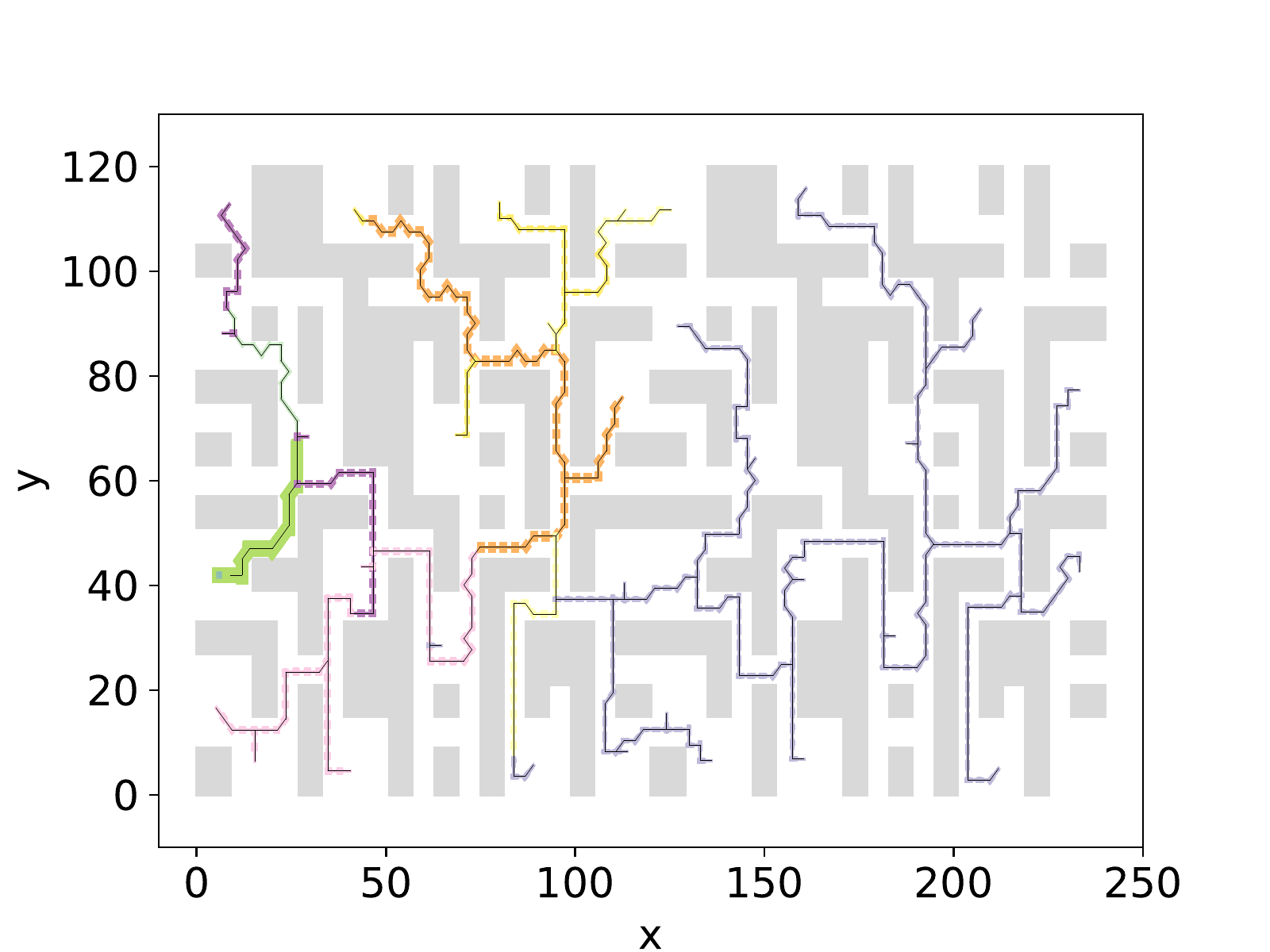}
    \includegraphics[width=0.49\textwidth]{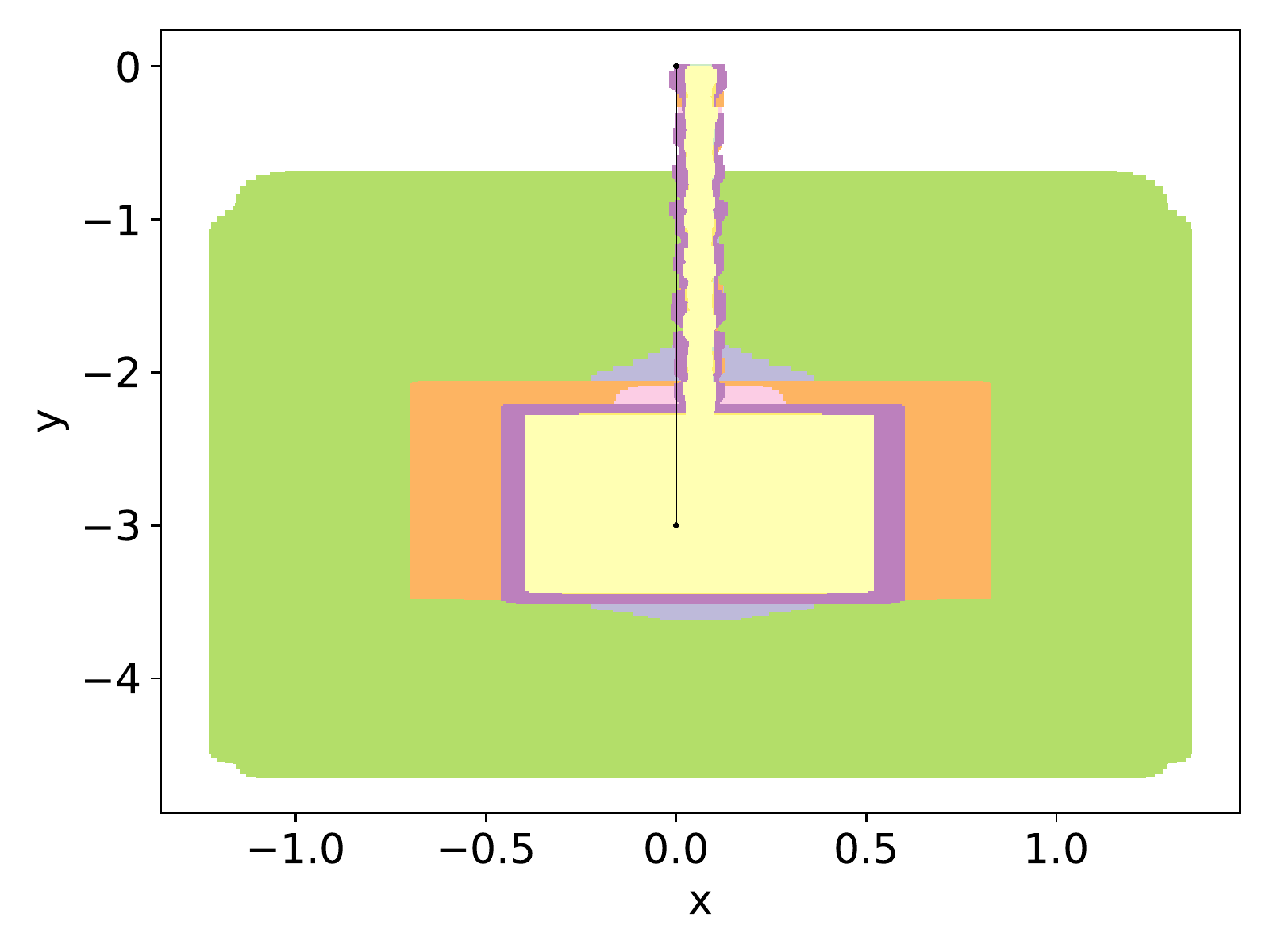}
    \includegraphics[width=0.49\textwidth]{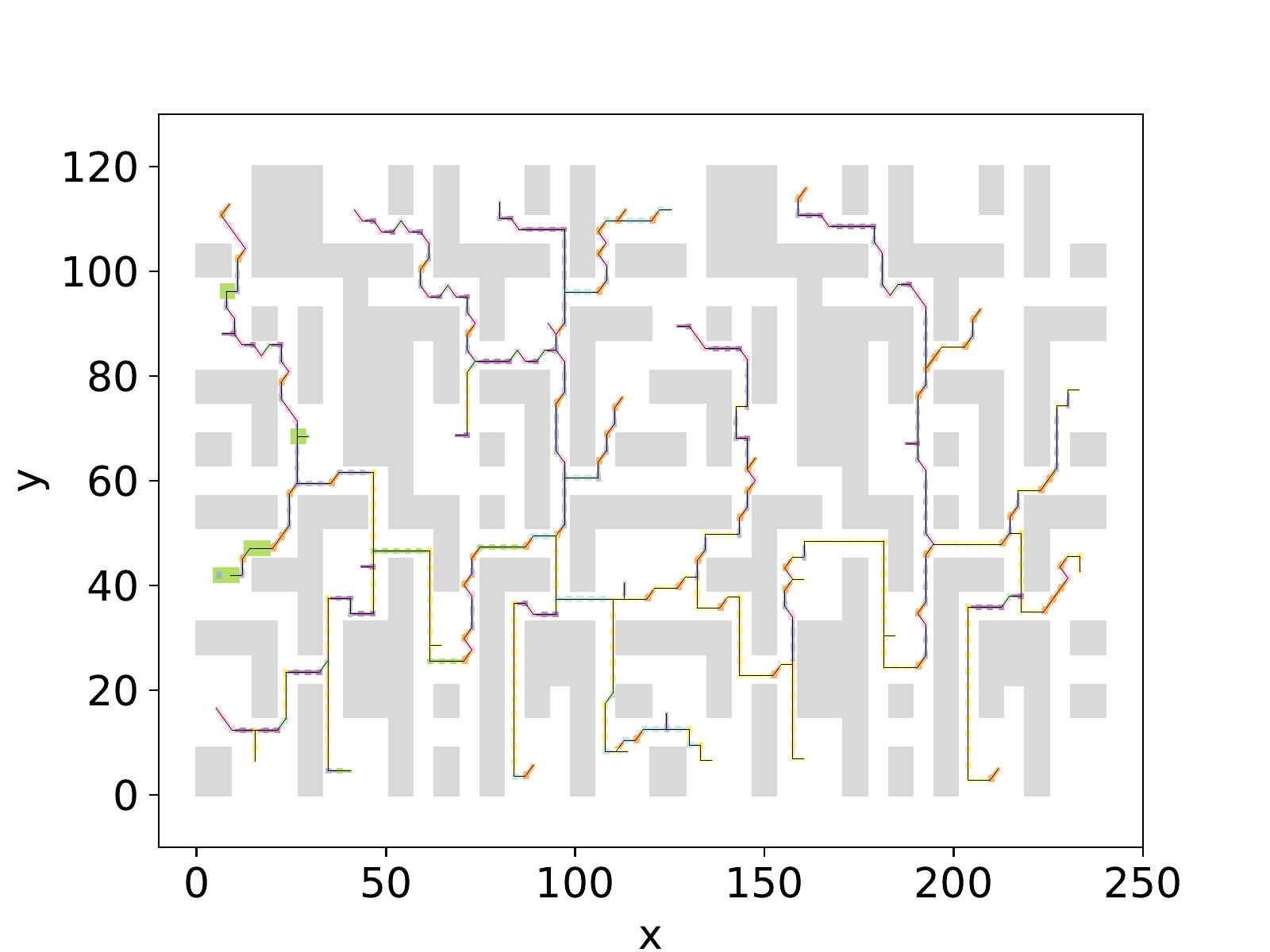}
    \includegraphics[width=0.49\textwidth]{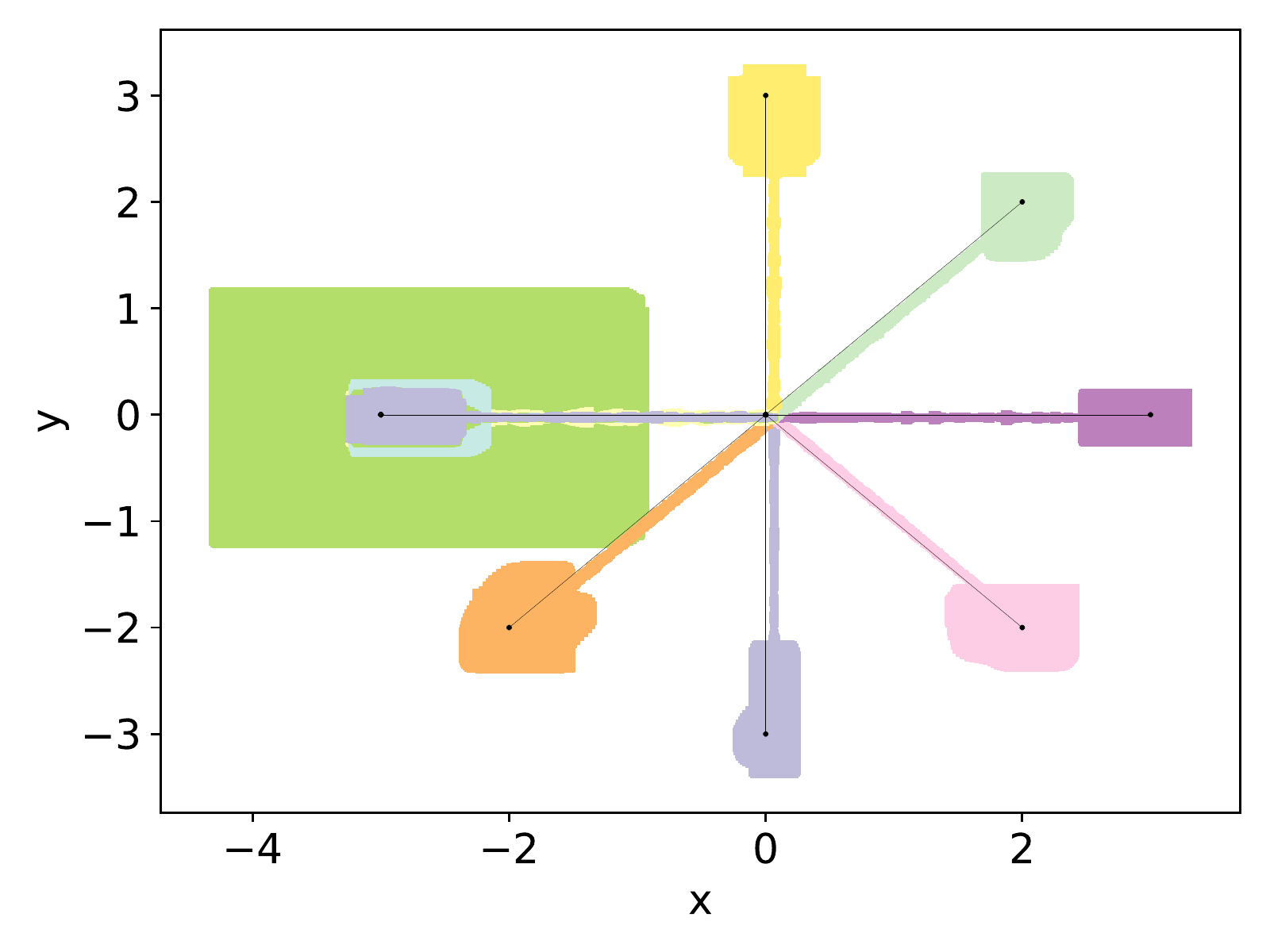}
  \end{center}
  \caption{Reachsets of the quadrotor with the NN-controller while using different $\Phi$s. The first row is the quadrotor's reachset when using $\Phi = $ TR. The second row is when using $\Phi = $ T. The left column represents the concrete scenario with the computed reachsets. The right column represents the reachsets of the abstract automaton.\label{fig:summary} }
  \label{fig:symmery_type}
\end{figure}

\end{document}